\documentclass[reqno]{amsart}

\usepackage{mathtools,amsmath,amsthm,amssymb,amsfonts,enumerate,graphicx,color,pstricks}
\usepackage[T1]{fontenc} 
\numberwithin{equation}{section} 
\theoremstyle{plain}
\newtheorem{theo+}           {Theorem}      [section]
\newtheorem{prop+}  [theo+]  {Proposition}
\newtheorem{coro+}  [theo+]  {Corollary}
\newtheorem{lemm+}  [theo+]  {Lemma}
\newtheorem{defi+}  [theo+]  {Definition}
\newtheorem{conj+}  [theo+]  {Conjecture}

\theoremstyle{definition}
\newtheorem{rema+}  [theo+]  {Remark}
\newtheorem{prob+}  [theo+]  {Problem}
\newtheorem{exam+}  [theo+]  {Example}

\newenvironment{theorem}{\begin{theo+}}{\end{theo+}}
\newenvironment{proposition}{\begin{prop+}}{\end{prop+}}
\newenvironment{corollary}{\begin{coro+}}{\end{coro+}}
\newenvironment{lemma}{\begin{lemm+}}{\end{lemm+}}

\newcommand{\om}{\omega}
\newcommand{\tha}{\theta}
\newcommand{\ti}{\textup i}
\newcommand{\sgn}{\operatorname{sgn}}
\newcommand{\id}{\operatorname{id}}
\newcommand{\pfaff}{\mathop{\mathrm{pfaff}}}
\newcommand{\trig}{{\operatorname{trig}}}
\newcommand{\Ker}{{\operatorname{Ker}}}

\begin{document}

\baselineskip 18pt
\larger[2]
\title[Special polynomials related to the eight-vertex model I]
{Special polynomials related to the supersymmetric eight-vertex model. I. Behaviour at cusps.} 
\author{Hjalmar Rosengren}
\address
{Department of Mathematical Sciences
\\ Chalmers University of Technology and University of Gothenburg\\SE-412~96 G\"oteborg, Sweden}
\email{hjalmar@chalmers.se}
\urladdr{http://www.math.chalmers.se/{\textasciitilde}hjalmar}

\thanks{Research  supported by the Swedish Science Research
Council (Vetenskapsr\aa det)}

\begin{abstract}
We study  certain symmetric polynomials, which as very special cases include polynomials related to the supersymmetric eight-vertex model, and other elliptic lattice models with  $\Delta=\pm 1/2$. In this paper, which is the first part of a series, we study the behaviour of the polynomials at special parameter values, which can be identified with cusps of the modular group $\Gamma_0(12)$.
In subsequent papers, we will show that the polynomials
satisfy a non-stationary Schr\"odinger equation related to the Knizhnik--Zamolodchikov--Bernard equation and that they give a four-dimensional lattice of  tau functions of Painlev\'e VI. 
\end{abstract}

\maketitle        


\section{Introduction}

The present paper is the first in a series, 
devoted to the study of certain special symmetric polynomials.
Numerous specializations of these polynomials have appeared in connection with elliptic solvable lattice models, at
the special parameter values usually denoted  $\Delta=\pm 1/2$.

The  values $ \Delta=\pm 1/2$ are truly exceptional.
One explanation why  $\Delta=-1/2$ is special comes from the limit to the massive sine-Gordon model, when it becomes a condition for supersymmetry \cite{fs}. Recently, Hagendorf and Fendley \cite{hf} implemented this supersymmetry on the finite lattice. Thus, we  refer to  
$ \Delta=-1/2$ as the \emph{supersymmetric} case. 

The six-vertex model with $\Delta=1/2$ contains the combinatorial ice model, where all states have equal weight. This was used by Kuperberg \cite{ku} in his simple proof of the alternating sign matrix theorem, which enumerates the states with domain wall boundary conditions.
The combinatorics of the  XXZ  and six-vertex models  at $\Delta=-1/2$ is also very rich, see \cite{zs} for a  survey. 
It seems quite interesting
  to extend results in this area to the elliptic regime, but so far only the first steps have been taken. 

In \cite{bm1}, Bazhanov and Mangazeev found that the ground state eigenvalue of Baxter's $Q$-operator for the supersymmetric periodic XYZ chain of odd length can be expressed in terms of certain polynomials, which  appear to have positive integer coefficients and thus call for a combinatorial interpretation. In \cite{bm2}, it was conjectured that specializations of these polynomials are tau functions of  Painlev\'e VI, constructed from one of Picard's algebraic   solutions via a sequence of B\"acklund transformations. 
The papers \cite{bm4} and \cite{ras}
 deal with ground state eigenvectors of the Hamiltonian for the same XYZ  chain. Certain components of these eigenvectors, as well as certain sums of components, again seem to be described by polynomials with positive coefficients. The same polynomials  appear for
other supersymmetric spin chains \cite{bh,fh,h}. A mathematically rigorous investigation of the  supersymmetric XYZ  chain
was recently initiated by  Zinn-Justin \cite{zj}.

As was noted in  \cite{bm4}, there are striking parallels between the
work outlined above and our  investigation of the 8VSOS and three-colour models \cite{r0,r}. Just as the six-vertex  model contains the combinatorial ice model when $\Delta=1/2$, the corresponding combinatorial specialization of the
 8VSOS model is the three-colour model.
Extending Kuperberg's work to the elliptic regime, we expressed the domain wall partition function for the three-colour model in terms of certain special polynomials, which again conjecturally have positive integer coefficients.

In the present paper, we will explain the relations between various polynomials introduced in \cite{bm1,bm4,r,zj}, by identifying them as special cases of
a more general family of functions.  
We stress that, although  the underlying physical models are  closely related \cite{b,bv}, it is not clear why  objects as different as domain wall partition functions  \cite{r0,r}, eigenvalues of the $Q$-operator  \cite{bm1,bm2} and eigenvectors of the Hamiltonian  \cite{bh,fh,h,bm4,ras,zj} should  lead to related special functions.

This first part in our series  has a rather technical nature.
For each non-negative integer $m$, we define a four-dimensional
lattice  $T_n^{(\mathbf k)}$ of symmetric rational functions 
in $m$ variables,
depending also on a parameter $\zeta$. 
The indices $n\in \mathbb Z$ and $\mathbf k=(k_0,k_1,k_2,k_3)\in\mathbb Z^4$ satisfy $|\mathbf k|+m=2n$. 
Since the denominator  is elementary,  $T_n^{(\mathbf k)}$ are essentially symmetric polynomials. One of our main results
is Corollary~\ref{tnvc}, which states that   $T_n^{(\mathbf k)}$ never vanishes identically. This  is  the key for our continued investigations \cite{r1b,r2}. 

Up to a change of variables, 
 $T_n^{(\mathbf k)}$ is an elliptic  function with periods $(1,\tau)$. As a function of $\tau$, the parameter $\zeta$ is a Hauptmodul for the modular group $\Gamma=\Gamma_0(6,2)\simeq \Gamma_0(12)$, see \S \ref{ms}. 
This group has six cusps. Under the natural action of
 $\mathrm N_{\mathrm{SL}(2,\mathbb Z)}(\Gamma)/\Gamma\simeq \mathrm{S}_3$,  the cusps split into two orbits. We call the three cusps belonging to the orbit of
$\tau=\ti\infty$ \emph{trigonometric} and the remaining three cusps  
\emph{hyperbolic}. 
Our proof of the fundamental Corollary \ref{tnvc} is based on a careful investigation of the behaviour of 
 $T_n^{(\mathbf k)}$
at the trigonometric cusps, see
 \S \ref{tls}. In \S \ref{ntcs}, we make a corresponding investigation of the hyperbolic cusps.  In the context of the XYZ spin chain,
the trigonometric cusps correspond to reductions
to the XXZ chain and the hyperbolic cusps 
 to the XY chain, see \S \ref{ms}.

In the final \S \ref{cns}, we explain the relation between  $T_n^{(\mathbf k)}$
and various polynomials introduced in \cite{bm1,bm4,r,zj}
and also occurring in \cite{bh,bm2,fh,h,ras}.
To be precise, for the polynomials related to eigenvalues of
the $Q$-operator \cite{bm1} and to domain wall partition functions \cite{r}, these relations are rigorously proved. However, for polynomials related to eigenvectors of the Hamiltonian, the identification with our polynomials 
$T_n^{(\mathbf k)}$ is still based  on empirical observation.
A partially rigorous result exists only for the
  ``sum rule'' giving the square norm of the eigenvector, which was
recently proved by  Zinn-Justin \cite{zj}, assuming a certain conjecture.

In the next paper in the series \cite{r1b}, we will obtain a non-stationary Schr\"odinger equation for $T_n^{(\mathbf k)}$, which is related to the  Knizhnik--Zamolodchikov--Bernard equation and to  the canonical quantization of Painlev\'e VI. 
This is in turn applied in \cite{r2}, where  the case $m=0$ of 
$T_n^{(\mathbf k)}$   is identified with a four-dimensional lattice of tau functions of Painlev\'e VI. 
In very special cases, these results reduce to conjectures of Bazhanov and Mangazeev \cite{bm1,bm2}.
In future work  we plan to explain the connection  to affine Lie algebra characters, which was already suggested in \cite{r}, and give further relations to the combinatorics of three-colourings. 
We hope that this  will  
be useful for understanding the
 combinatorics of elliptic lattice models at $\Delta=-1/2$.
The resulting relation between Painlev\'e tau functions and 
affine Lie algebra characters should be of independent interest.

{\bf Acknowledgements:} I would like to thank Vladimir Bazhanov, Vladimir Mangazeev and Paul Zinn-Justin for discussions.

\section{Preliminaries}

\subsection{Notation}  
Throughout the paper,
$$\omega=e^{2\pi\ti/3}. $$

We fix
 $\tau$ in the upper half-plane, and write
 $p=e^{\pi\ti\tau}$. 
The four half-periods in  $\mathbb C/(\mathbb Z+\tau\mathbb Z)$ will be denoted
$$
\gamma_0=0,\qquad \gamma_1=\frac\tau 2,\qquad \gamma_2=\frac\tau2+\frac 12,\qquad \gamma_3=\frac 12. $$

 We will use the notation
$$(x;p)_\infty=\prod_{j=0}^\infty(1-xp^j),$$
$$\theta(x;p)=(x;p)_\infty(p/x;p)_\infty.$$
Repeated variables are used as a short-hand for products; for instance,
$$\theta(a,\pm b;p)=\theta(a;p)\theta(b;p)\theta(-b;p). $$
The theta function satisfies 
\begin{equation}\label{tqp}\theta(px;p)= \theta(x^{-1};p)=-x^{-1}\theta(x;p)
\end{equation}
and the addition formula
\begin{equation}\label{tad}\tha(az^\pm,bw^\pm;p)
-\tha(bz^\pm,aw^\pm;p)
=\frac bz\,\tha(ab^\pm,zw^\pm;p).\end{equation}

We introduce the function
 $$\psi(z)=p^{\frac 1{12}}(p^2;p^2)_\infty e^{-\pi \ti z}\theta(e^{2\pi \ti z},\pm pe^{2\pi \ti z};p^2),$$
which satisfies
\begin{subequations}
\label{pss}
\begin{equation}
 \psi(z+1)=\psi(-z)=-\psi(z),\qquad
\psi(z+\tau)={e^{-3\pi \ti(\tau+2z)}}\psi(z).
\end{equation}
By the quintuple product identity \cite{w}, 
$$\psi(z)=\sum_{k=-\infty}^\infty \left(e^{\frac 1{12}(6k-1)^2\pi \ti\tau+(6k-1)\pi \ti z}
-e^{\frac 1{12}(6k+1)^2\pi \ti\tau+(6k+1)\pi \ti z}
\right),$$
which implies that 
 \begin{equation}\psi(z)=\psi\left(z+\frac 13\right)+\psi\left(z-\frac 13\right). \end{equation}
\end{subequations}

\subsection{Spaces of theta functions} \label{sts}

For $n$ a non-negative integer, we denote
 by $\Theta_n$ the space of entire functions $f$ such that
\begin{subequations}\label{vde}
\begin{equation}\label{fqp}f(z+1)=f(z),\qquad  f(z+\tau)=e^{-6\pi \ti n(\tau+2z)}f(z),\qquad f(-z)=-f(z),\end{equation}
\begin{equation}\label{fe}f(z)+f\left(z+\frac 13\right)+f\left(z-\frac 13\right)=0. \end{equation}
\end{subequations}
For $n\geq 1$, 
 the map $e^{2\pi\ti z}\mapsto f(z)$ is,
in the terminology of \cite[Def.\ 3.1]{rs}, a $C_{3n-1}$ theta function, subject to the additional constraint  \eqref{fe}. By \cite[Prop.\ 6.1]{rs}, $\dim \Theta_n=2n$, an explicit basis being
$$e^{2\pi\ti(j-3n)z}\theta(-p^{2j}e^{12\pi\ti n z};p^{12n})-e^{2\pi\ti(3n-j)z}\theta(-p^{2j}e^{-12\pi\ti n z};p^{12n}),$$
where $1\leq j\leq 3n-1$ and $3\nmid j$.
The space $\Theta_0$ consists only of the zero function. 




\begin{lemma}\label{tsl}
The space $\Theta_n$ is the linear span of all functions of the form
\begin{equation}\label{tss}e^{-2\pi\ti z}\tha(e^{4\pi \ti z},ae^{\pm 2\pi\ti z};p^2)\tha(b_1e^{\pm 6\pi\ti z},\dots,b_{n-1}e^{\pm 6\pi\ti z};p^6), \end{equation}
where $a,b_1,\dots,b_{n-1}$ are constants.
\end{lemma}

\begin{proof}
Starting from \eqref{tss}, we may use
 \eqref{tad} to express
$\tha(ae^{\pm 2\pi\ti z};p^2)$ as a linear combination of
$\tha(\om e^{\pm 2\pi\ti z};p^2)$ and
$\tha(-\om e^{\pm 2\pi\ti z};p^2)$.
For $a=\omega$, \eqref{tss} is a constant times
$$e^{-3\pi\ti z}\psi\left(z+\frac 12\right)\tha(e^{ 6\pi\ti z},b_1e^{\pm 6\pi\ti z},\dots,b_{n-1}e^{\pm 6\pi\ti z};p^6)$$
and for $a=-\omega$ a constant times
$$e^{-3\pi\ti z}\psi(z)\tha(-e^{6\pi\ti z},b_1e^{\pm 6\pi\ti z},\dots,b_{n-1}e^{\pm 6\pi\ti z};p^6).$$
 Using \eqref{tqp} and \eqref{pss} it is straight-forward to  check that these functions are in  $\Theta_n$. To complete the proof, it is enough to find $2n$ linearly independent functions $\psi_1,\dots,\psi_{2n}$ of the form \eqref{tss}. 
For later purposes, it will be convenient to define 
\begin{align}
\notag\psi_j(z)&=e^{-2\pi\ti z}\tha(e^{4\pi \ti z};p^2)\tha(e^{\pm 2\pi\ti z};p^2)^{\left[\frac{j-1}2\right]}\tha(\om e^{\pm 2\pi\ti z};p^2)^{j-1}\\
\notag &\qquad\qquad\times\tha(p e^{\pm 2\pi\ti z};p^2)^{n-\left[\frac{j+1}2\right]}\tha(p\om e^{\pm 2\pi\ti z};p^2)^{2n-j}\\
\label{fpj}&=\begin{cases}
\om^{j-2}e^{-2\pi\ti z}\tha(e^{4\pi \ti z},\om e^{\pm 2\pi\ti z};p^2)\\
\qquad  \times\tha(e^{\pm 6\pi\ti z};p^6)^{\frac{j-2}2}\tha(p^3e^{\pm 6\pi\ti z};p^6)^{\frac{2n-j}2}, & j \text{ even},\\
\om^{j-1}e^{-2\pi\ti z}\tha(e^{4\pi \ti z},\om p e^{\pm 2\pi\ti z};p^2)\\
\qquad\times\tha(e^{\pm 6\pi\ti z};p^6)^{\frac{j-1}2}\tha(p^3e^{\pm 6\pi\ti z};p^6)^{\frac{2n-j-1}2}, & j \text{ odd},
\end{cases}
\end{align}
which are linearly independent since
 $\psi_j$ has a zero of degree exactly $j-1$ at $z=1/3$.
\end{proof}

We introduce the operator 
\begin{equation}\label{sig}(\sigma f)(z)=\frac{\ti}{\sqrt 3}\left(f\left(z+\frac 13\right)-f\left(z-\frac 13\right)\right). \end{equation}
It is easy to check that $\sigma$ acts as an involution 
 on the  space  of 
 one-periodic functions satisfying \eqref{fe}. 
Moreover, $\sigma$ preserves the second relation in \eqref{fqp}, whereas it 
interchanges the subspaces of even and odd functions. Thus, $\sigma$ is a bijection from $\Theta_n$ to the space of entire functions  satisfying \eqref{vde}, except that they are even instead of odd.


We will now introduce a generalization of the space $\Theta_n$, defined by 
combining \eqref{vde} with a prescribed behaviour at the lattice 
\begin{equation}\label{lal}\Lambda=\frac 16\,\mathbb Z+\frac\tau 2\,\mathbb Z.
\end{equation}
It will  be convenient to write  $\Lambda_j=\gamma_j+\frac 13\mathbb Z+\tau\mathbb Z$, so that $\Lambda=\bigcup_{j=0}^3\Lambda_j$.
 To prepare for the definition, consider a meromorphic function $f$ satisfying \eqref{vde}, where $n$ is no longer assumed positive. 
Let  $g_j(z)=f(z+\gamma_j)$ if $j=0,3$ and  
$g_j(z)=e^{6\pi\ti nz}f(z+\gamma_j)$ if $j=1,2$. It is easy to check that
the functions $g_j$ also satisfy \eqref{vde}.

Let $n\in\mathbb Z$ and $\mathbf k=(k_0,k_1,k_2,k_3)\in\mathbb Z^4$,
such that $2n\geq|\mathbf k|=\sum_jk_j$. Throughout, we will write
$m=2n-|\mathbf k|$ and $k_j^\pm =\max(\pm k_j,0)$.
We define $\Theta_{n}^{\mathbf k}$ to be the space of  functions $f$
satisfying \eqref{vde}, which are analytic except for possible poles at
 $\Lambda$, and such that 
 the Laurent expansion of $g_j(z)$ at $z=1/3$ takes the form
\begin{equation}\label{ocd}g_j\left(z+\frac 13\right)=\sum_{N=0}^\infty a_N z^{2N}+\sum_{N=k_j}^\infty b_N z^{2N+1}, \qquad j=0,1,2,3.\end{equation}

\begin{lemma} \label{kl}
With $\mathbf k$ and $n$ as above, let
$$\Phi(z)=\prod_{j=0}^3\theta(e^{6\pi\ti (\gamma_j\pm z)};p^6)^{k_j^-}. $$
Then, $f\in \Theta_{n}^{\mathbf k}$ if and only if $F=f\Phi$ is an element in $\Theta_{n+|\mathbf k^-|}$ 
satisfying  
\begin{subequations}
\begin{equation}\label{fva}F'(\gamma_j)=F^{(3)}(\gamma_j)=\dots=F^{(2k_j-1)}(\gamma_j)=0 \end{equation}
 for indices $j$ such that $k_j> 0$ and 
\begin{equation}\label{fvb}\sigma F(\gamma_j)=(\sigma F)^{''}(\gamma_j)=\dots=(\sigma F)^{(-2k_j-2)}\left(\gamma_j\right)=0
\end{equation}
\end{subequations}
 for $j$ such that $k_j<0$.
\end{lemma}

\begin{proof}
It is easy to check that $f$  satisfies \eqref{vde}
if and only if $F$  satisfies \eqref{vde} with $n$ replaced by
$n+|\mathbf k^-|$. Thus, it is enough to consider the behaviour at $\Lambda$.
Note  that  \eqref{ocd} implies
\begin{subequations}\label{gt}
\begin{equation}g_j(z)=-g_j\left(z+\frac 13\right)+g_j\left(-z+\frac 13\right)=-2\sum_{N=k_j}^\infty b_N z^{2N+1}.\end{equation}
Thus, 
if $k_j\geq 0$, \eqref{ocd} means that  $f$ is analytic at $\Lambda_j$
and has a zero of degree at least $2k_j+1$ at $z=\gamma_j$.
Since $g_j$ is odd, it is for the latter property enough to assume the 
vanishing of the first $k_j$ odd derivatives.
Since $\Phi$ is even and non-zero at $\Lambda_j$, 
this is in turn equivalent to 
$F$ being analytic at $\Lambda_j$ and satisfying \eqref{fva}.

If $k_j<0$,  $f$ has  poles of degree at most
$-(2k_j+1)$  at $\Lambda_j$; thus,  $F$ is analytic  there.
Writing 
\begin{equation}(\sigma g_j)(z)=\frac{\ti}{\sqrt 3}\left(g_j\left(z+\frac 13\right)+g_j\left(-z+\frac 13\right)\right)=
\frac{\ti}{\sqrt 3}\sum_{N=0}^\infty a_N z^{2N},\end{equation}
\end{subequations}
we find as above that \eqref{ocd} is equivalent to $F$ being analytic at 
$\Lambda_j$ and satisfying \eqref{fvb}.
\end{proof}

The following alternative characterization of the space $\Theta_n^{\mathbf k}$ 
will be used in \cite{r1b}. It is an easy consequence of \eqref{gt}.

\begin{lemma}
The space $\Theta_n^{\mathbf k}$ equals 
the space of functions $f$
satisfying \eqref{vde}, which are analytic except for possible poles at
 $\Lambda$, and such that, for $j=0,1,2,3$, 
$$\lim_{z\rightarrow \gamma_j}(z-\gamma_j)^{1-2k_j}f(z)=\lim_{z\rightarrow\gamma_j}(z-\gamma_j)^2\left(f\left(z+\frac 13\right)-f\left(z-\frac 13\right)\right)=0. $$
\end{lemma}

Since,
by Lemma \ref{kl}, $\Theta_{n}^{\mathbf k}$ may be obtained by imposing $\sum_j|k_j|=2|\mathbf k^-|+|\mathbf k|$ linear conditions on a space of dimension $2n+2|\mathbf k^-|$, its dimension is at least
$2n-|\mathbf k|=m$. In fact,
 this estimate is  sharp.


\begin{theorem}\label{dt}
The space $\Theta_{n}^{\mathbf k}$ has dimension $m=2n-|\mathbf k|$.
\end{theorem}

Theorem \ref{dt} is 
vital for our continued investigations \cite{r1b,r2}.
It is proved at the end of \S \ref{tls}.

\subsection{Elliptic Tsuchiya determinant}

We are interested in the one-dimensional space $\Theta_{n}^{\wedge 2n}$. It is
spanned by the alternant 
\begin{equation}\label{alt}\det\left(\psi_j(z_i)\right)_{1\leq i,j\leq 2n},
\end{equation}
 where $\psi_j$ runs through a basis of  $\Theta_{n}$. 
We will also need another type of determinant formula, which is less symmetric, but has the advantage that any minor of the determinant involved is a determinant of the same form.

\begin{lemma}
The space $\Theta_{n}^{\wedge 2n}$ is spanned by the function
\begin{multline}\label{eik}\prod_{j=1}^{2n}e^{-2\pi\ti z_j}\tha(e^{4\pi \ti z_j};p^2)
\prod_{i,j=1}^ne^{-6\pi\ti z_{n+j}}\tha(e^{6\pi \ti (z_{n+j }\pm z_i)};p^6)\\
\times\det_{1\leq i,j\leq n}\left(\frac{e^{-2\pi\ti z_{n+j}}\tha(e^{2\pi \ti (z_{n+j }\pm z_i)};p^2)}{e^{-6\pi\ti z_{n+j}}\tha(e^{6\pi \ti (z_{n+j }\pm z_i)};p^6)}\right).
 \end{multline}
\end{lemma}

(When $n=0$, \eqref{eik} should be interpreted as the constant  $1$, and $\Theta_0^{\wedge 0}$ as $\mathbb C$.)

\begin{proof}
As a function of $z_1$, \eqref{eik} is a linear combination of the terms
$$e^{-2\pi\ti z_1}\tha(e^{4\pi \ti z_1};p^2)\tha(e^{2\pi \ti (z_{n+k }\pm z_1)})\prod_{j=1,\,j\neq k}^n\tha(e^{6\pi \ti (z_{n+j }\pm z_1)};p^6), $$
with $1\leq k\leq n$, which  belong to $\Theta_n$ by Lemma \ref{tsl}.

Next, we prove that \eqref{eik} is  antisymmetric in $z_1,\dots,z_{2n}$.
By \eqref{tad}, 
\begin{multline*}\prod_{i,j=1}^ne^{-6\pi\ti z_{n+j}}\tha(e^{6\pi \ti (z_{n+j }\pm z_i)};p^6)\det_{1\leq i,j\leq n}\left(\frac{e^{-2\pi\ti z_{n+j}}\tha(e^{2\pi \ti (z_{n+j }\pm z_i)};p^2)}{e^{-6\pi\ti z_{n+j}}\tha(e^{6\pi \ti (z_{n+j }\pm z_i)};p^6)}\right)\\
\sim\prod_{j=1}^{2n}{\tha(be^{\pm 2\pi\ti z_j};p^2)}{\tha(b^3e^{\pm 6\pi\ti z_j};p^6)^{n-1}}
\prod_{i,j=1}^n\left(v_{n+j}-v_i\right)
\det_{1\leq i,j\leq n}\left(\frac{u_{n+j}-u_i}{v_{n+j}-v_i}\right),
\end{multline*}
where 
\begin{equation}\label{uv}u_j=\frac{\tha(ae^{\pm 2\pi\ti z_j};p^2)}{\tha(be^{\pm 2\pi\ti z_j};p^2)},\qquad v_j=\frac{\tha(a^3e^{\pm 6\pi\ti z_j};p^6)}{\tha(b^3e^{\pm 6\pi\ti z_j};p^6)}, \end{equation}
with $a$ and $b$ arbitrary parameters with $ab,\,a/b\notin p^{2\mathbb Z}$. 
Here and below, $\sim$ means equality up to a non-zero multiplicative factor independent of the variables $z_j$.
The antisymmetry now follows from the general fact that 
\begin{equation}\label{od}\prod_{i,j=1}^n\left(v_{n+j}-v_i\right)
\det_{1\leq i,j\leq n}\left(\frac{u_{n+j}-u_i}{v_{n+j}-v_i}\right) \end{equation}
is antisymmetric under simultaneous permutations of $u_j$ and $v_j$, see \cite[Thm.\ 4.2]{o}.

It remains to show that \eqref{eik} is not identically zero. To this end, we consider the limit $z_{2n}\rightarrow z_n+1/3$. Expanding the determinant along the last row, only the entry in the last column contributes to the limit. Therefore,  \eqref{eik} reduces to a non-zero factor times a similar expression with $n$ replaced by $n-1$. The non-vanishing of \eqref{eik} then follows by induction on $n$. 
\end{proof}

When $p=0$, \eqref{eik} is essentially  the Tsuchiya determinant \cite{ts} (with $\Delta=1/2$), which is the  partition function of the  six-vertex model on a rectangle bounded by one reflecting edge and three domain walls.  
Recently, Filali generalized this to the elliptic level, interpreting
 \eqref{eik} as a partition function for  the 8VSOS model
\cite{f}.

The antisymmetry of \eqref{od} is explicitly displayed in the identity
\cite[Thm.\ 7.2]{l}
\begin{multline}\label{li}\prod_{i,j=1}^n\left(v_{n+j}-v_i\right)
\det_{1\leq i,j\leq n}\left(\frac{u_{n+j}-u_i}{v_{n+j}-v_i}\right)\\
=\prod_{1\leq i<j\leq 2n}(\sqrt{v_j}+\sqrt{v_i})\pfaff_{1\leq i,j\leq 2n}\left(\frac{u_{j}-u_i}{\sqrt{v_{j}}+\sqrt{v_i}}\right),
\end{multline}
see also \cite{rcd}. In the case at hand, it is natural to
 choose $a$ and $b$ in \eqref{uv} as  two of the four numbers 
$e^{2\pi\ti\gamma_j}$, $0\leq j\leq 3$, since $v_j$ then have meromorphic square roots. 
For instance, choosing  $a=1$, $b=-1$ yields that \eqref{eik} is equal to 
\begin{multline*}\frac {\prod_{j=1}^{2n}e^{-2\pi\ti(3n-2) z_j}\tha(e^{4\pi \ti z_j};p^2)}{\tha(-1;p^6)^{2n(n-1)}}\\
\times\prod_{1\leq i<j\leq 2n}\Big(\tha(e^{6\pi \ti z_j},-e^{6\pi \ti z_i};p^6)
+\tha(-e^{6\pi \ti z_j},e^{6\pi \ti z_i};p^6)\Big)\\
\times\pfaff_{1\leq i,j\leq 2n}\left(\frac{e^{-2\pi\ti z_{j}}\tha(e^{2\pi \ti (z_{j }\pm z_i)})}{\tha(e^{6\pi \ti z_j},-e^{6\pi \ti z_i};p^6)
+\tha(-e^{6\pi \ti z_j},e^{6\pi \ti z_i};p^6)}\right).
\end{multline*}
 There exist $\binom{4}{2}=6$ pfaffian formulas of this nature. Although 
more symmetric than \eqref{eik}, they seem less useful for our purposes.

\subsection{Elliptic functions}

We will write
$$x(z)=x(z,\tau)=\frac{\theta(- p\omega;p^2)^2\theta(\omega e^{\pm 2\pi \ti z};p^2)}{\theta(-\omega;p^2)^2\theta( p\omega e^{\pm 2\pi \ti z};p^2)} $$
(in \cite{r}, $x(z)$ was denoted $\xi(e^{2\pi \ti z})$).
This is an even elliptic function with periods $1$ and $\tau$, and in fact generates the field of all such functions. 
By \eqref{tad},
\begin{equation}\label{xd} x(z)- x(w)\sim\frac{\tha(e^{2\pi \ti(z\pm w)};p^2)}{\tha( p\om e^{\pm 2\pi \ti z}, p\om e^{\pm 2\pi \ti w};p^2)},\end{equation}
which implies
\begin{equation}\label{xp} x'(z)\sim\frac{e^{-2\pi \ti z}\tha(e^{4\pi \ti z};p^2)}{\tha( p\om e^{\pm 2\pi \ti z};p^2)^2}.\end{equation} 
The constants of proportionality (which depend on $\tau$) can be given explicitly but are irrelevant for our purposes.

As in \cite{r}, we will write
\begin{equation}\label{z}\zeta=\zeta(\tau)=\frac{\omega^2\theta(-1,- p\omega;p^2)}{\theta(- 
p,-\omega;p^2)}. \end{equation}
As we discuss in \S \ref{ms}, $\zeta$ generates the field of modular functions for the congruence group $\Gamma_0(6,2)\sim\Gamma_0(12)$. 
Thus, if a function of $(z,\tau)$ has the appropriate elliptic and modular behaviour, it is automatically a rational function of $(x,\zeta)$. We refer to the change of variables from $(z,\tau)$ to $(x,\zeta)$ as \emph{uniformization}.

The values of $x$ on the lattice $\Lambda$
are  rational expressions in $\zeta$.

\begin{lemma}\label{xvl}
The values $\xi_j=x(\gamma_j)$ and $\eta_j=x(\gamma_j+1/3)$ are given by
\begin{align*}
\xi_0&=2\zeta+1,& \xi_1&=\frac{\zeta}{\zeta+2},
& \xi_2&=\frac{\zeta(2\zeta+1)}{\zeta+2},
& \xi_3&=1,\\
\eta_0&=0, & \eta_1&=\infty, &\eta_2&=\frac{2\zeta+1}{\zeta+2},
& \eta_3&=\zeta.
\end{align*}
\end{lemma}

These evaluations are all contained in, or follow from,
\cite[Lemma 7.7]{r}.



Translating $z$ by half-periods corresponds to rational transformations of 
$x$.

\begin{lemma}\label{hpl}
If  $x=x(z)$, then
$$x\left(z+\frac 1 2\right)=\frac{(2\zeta+1)(x-\zeta)}{(\zeta+2)x-(2\zeta+1)}, 
\qquad x\left(z+\frac\tau 2\right)=\frac{\zeta(2\zeta+1)}{(\zeta+2)x}. $$
\end{lemma}

\begin{proof}
 By Lemma \ref{xvl}, both sides  of the two identities  are elliptic 
functions of $z$ 
with the same periods, zeroes and poles. 
Thus, it suffices to verify the case
$z=0$, which  again follows from  Lemma \ref{xvl}.
\end{proof}

Translating $z$ by $1/3$ corresponds to an algebraic transformation of 
 $x$.

\begin{lemma}\label{mll}
 If $x=x(z)$,  $x_+=x(z+1/3)$ and $x_-=x(z-1/3)$, then
\begin{subequations}\label{xpm}
\begin{align}
x+x_++x_-&=2\zeta+1, \\
\frac 1x+\frac 1 {x_+}+\frac 1{x_-}&=\frac{\zeta+2}{\zeta}. \end{align}
\end{subequations}
\end{lemma}

\begin{proof}
We observe that $x+x_++x_-$ is an elliptic function of $z$ with periods $1/3$ and $\tau$. Modulo these periods, the only possible pole is at $z=\tau/2$, where 
$x_+$ and $x_-$ have simple poles. However,  
$$x_+\left(\frac\tau 2+z\right)+x_-\left(\frac \tau 2+z\right)
=x\left(\frac \tau 2+\frac 1 3+z\right)+x\left(\frac\tau 2+\frac 13-z\right)$$ 
is an even function of $z$  
and  thus cannot have a simple pole at $z=0$. It follows that 
$x+x_++x_-$ is independent of $z$. The value can then be computed using
Lemma \ref{xvl}. The second equation is proved similarly.
\end{proof}





\subsection{Uniformization of $\Theta_n^{\wedge 2n}$.}

By Lemma  \ref{tsl} and \eqref{xd}, any $f\in\Theta_n$ can be written
\begin{equation}\label{put}f(z)=e^{-2\pi\ti z}\tha(e^{4\pi\ti z};p^2)\tha(\om pe^{\pm 2\pi\ti z};p^2)^{3n-2}P(x(z)),\end{equation}
where $P$ is a polynomial of degree at most $3n-2$.

As in \cite{r}, we let
\begin{gather}\notag
G(x,y)=\frac 1{y-x}\left|\begin{matrix}(\zeta+2)x-\zeta & (\zeta+2)y-\zeta\\
x^2(x-2\zeta-1) &y^2(y-2\zeta-1)
\end{matrix}\right|\\
\label{gd}=
(\zeta+2)xy(x+y)-\zeta(x^2+y^2)-2(\zeta^2+3\zeta+1)xy+\zeta(2\zeta+1)(x+y). \end{gather}
For simplifying many identities below, it is useful to know that
\begin{align*}G(x,\xi_0)&=2(\zeta+1)^2x^2,& G(x,\xi_1)&=\frac{2\zeta^2(\zeta+1)^2}{(\zeta+2)^2},\\ 
G(x,\xi_2)&=\frac{2\zeta^2\big((\zeta+2)x-(2\zeta+1)\big)^2}{(\zeta+2)^2}, & 
G(x,\xi_3)&=2(x-\zeta)^2.\end{align*}

By \eqref{xd} and Lemma \ref{xvl},
\begin{equation}\label{gu}
e^{-6\pi\ti z_{2}}\tha(e^{6\pi\ti (z_{2}\pm z_1)};p^6)\sim \tha(\om pe^{\pm 2\pi\ti z_1},\om pe^{\pm 2\pi\ti z_2};p^2)^3(x_2-x_1)G(x_1,x_2),
\end{equation}
where $x_j=x(z_j)$.
The uniformization of \eqref{tss} is then given by 
\begin{equation}\label{gpu}
(x-x(a))\prod_{k=1}^{n-1}(x-x(b_k))G(x,x(b_k))\end{equation}
in the sense that if $P(x)$ denotes this expression, then 
\eqref{put} is proportional to \eqref{tss}. In the same sense, \eqref{fpj} is uniformized by 
\begin{equation}\label{psb}
P_{j}(x)=x^{j-1}(x-2\zeta-1)^{\left[\frac{j-1}2\right]}\big((\zeta+2)x-\zeta\big)^{n-\left[\frac{j+1}2\right]}.\end{equation}
Thus, using the basis  \eqref{fpj} in \eqref{alt}, we find that
  $\Theta_n^{\wedge 2n}$
is spanned by
\begin{equation}\label{tut}\prod_{j=1}^{2n}e^{-2\pi\ti z_j}\tha(e^{4\pi\ti z_j};p^2)
\tha(\om pe^{\pm 2\pi\ti z_j};p^2)^{3n-2}\Delta(x_1,\dots,x_{2n})\,T(x_1,\dots,x_{2n}), \end{equation}
where $x_j=x(z_j)$, $\Delta(\mathbf x)=\prod_{i<j}(x_j-x_i)$ and
\begin{equation}\label{sdt}T(x_1,\dots,x_{2n})=\frac{\det_{1\leq i,j\leq 2n}\left(P_j(x_i)\right)}{\Delta(x_1,\dots,x_{2n})}
\end{equation}
(by convention, $T=1$ when $n=0$).
This is a symmetric polynomial in $2n$ variables, depending also as a polynomial on the parameter $\zeta$. As we will see in \S \ref{zss}, it is related to the polynomial $H_{2n}$ introduced in \cite{zj} by a  change of variables.
 
We now give the uniformization of the determinant \eqref{eik}.

\begin{lemma}
The   polynomial
$T$ is given by the alternative determinant formula
\begin{equation}\label{tn} T(x_1,\dots,x_{2n})
=\frac{\prod_{i,j=1}^nG(x_i,x_{n+j})}{\Delta(x_1,\dots,x_n)\Delta(x_{n+1},\dots,x_{2n})}\,\det_{1\leq i,j\leq n}\left(\frac{1}{G(x_i,x_{n+j})}\right).
 \end{equation}
\end{lemma}

\begin{proof}
Using \eqref{xd} and \eqref{gu}, it is straight-forward to check that
\eqref{eik} is proportional to \eqref{tn}. Thus, \eqref{sdt} and \eqref{tn}
are equal up to a multiplicative factor independent of the variables $x_j$. 
To see that this factor is $1$, we compute the highest homogeneous component of both sides. Since the leading term in  $P_j(x)$ is $(\zeta+2)^{n-[(j+1)/2]}x^{n+j-2}$, since $\sum_{j=1}^{2n}(n-[(j+1)/2])=n(n-1)$
and  $\det(x_i^{j-1})=\Delta(\mathbf x)$, the highest homogeneous component of
\eqref{sdt} is $(\zeta+2)^{n(n-1)}\prod_{j=1}^{2n}x_j^{n-1}$. On the other hand,  
since
$$\lim_{t\rightarrow 0}t^3 G(x/t,y/t)=(\zeta+2)xy(x+y) $$
the highest homogeneous component of \eqref{tn} is
\begin{multline*}\frac{(\zeta+2)^{n(n-1)}\prod_{j=1}^{2n}x_j^{n-1}\prod_{i,j=1}^n(x_i+x_{n+j})}{\Delta(x_1,\dots,x_n)\Delta(x_{n+1},\dots,x_{2n})}
\det_{1\leq i,j\leq n}\left(\frac{1}{x_i+x_{n+j}}\right)\\
=(\zeta+2)^{n(n-1)}\prod_{j=1}^{2n}x_j^{n-1},
 \end{multline*}
by the Cauchy determinant evaluation.

\end{proof}

By \eqref{li},   
there are also  pfaffian formulas for $T$.
For instance, 
\begin{multline*}T(x_1,\dots,x_{2n})
=\prod_{1\leq i<j\leq 2n}\frac{\sqrt{f(x_i)h(x_j)}+\sqrt{f(x_j)h(x_i)}}{x_j-x_i}\\
\times\pfaff_{1\leq i,j\leq 2n}\left(\frac{x_j-x_i}{\sqrt{f(x_i)h(x_j)}+\sqrt{f(x_j)h(x_i)}}\right), \end{multline*}
where
$$f(x)=(\zeta+2)x-\zeta,\qquad h(x)=x^2(x-2\zeta-1). $$
Up to a multiplicative constant, this remains true if 
$f$ and $h$ are replaced by any two linearly independent 
elements in their span.

\subsection{Uniformization of the map $\sigma$}

As was explained in \S \ref{sts}, the map $\sigma$ defined in \eqref{sig} maps
$\Theta_n$ to the space of entire functions satisfying \eqref{vde}, {except} that they are even rather than odd. It follows that, if $f\in\Theta_n$,
\begin{equation}\label{q}(\sigma f)(z)=\theta(\omega pe^{\pm 2\pi\ti z};p^2)^{3n} Q(x(z)) \end{equation}
for some polynomial $Q$ of degree at most $3n$. 
The following result gives a useful description of the relation between $P$ and $Q$.

\begin{lemma}\label{spl}
For fixed $n$, define the linear operator $\hat \sigma$ 
on the span of $(P_j)_{j=1}^{2n}$ by
\begin{equation*}\hat\sigma( P_j)=\begin{cases}e_0(x)P_j(x), & j \text{ odd},\\
e_1(x)P_j(x), & j \text{ even},
\end{cases} \end{equation*}
where
\begin{align*}
e_0(x)&=\big((\zeta+2)x-\zeta\big)(2\zeta+1-3x),\\
e_1(x)&=(x-2\zeta-1)\big((\zeta+2)x-3\zeta\big).
\end{align*}
Then, there exists a constant $C=C(\tau)$ such that the polynomials $P$ and $Q$ in \eqref{put} and \eqref{q} are related by
$Q=C\hat\sigma( P)$. Moreover,
$$\hat\sigma\left((x-a)\prod_{k=1}^{n-1}(x-b_k)G(x,b_k)\right)
=(xe_1(x)-ae_0(x))\prod_{k=1}^{n-1}(x-b_k)G(x,b_k).
 $$
\end{lemma}

\begin{proof}
Writing $Q=\tilde\sigma(P)$,
we must prove that $\tilde\sigma=C\hat\sigma$.
 It is clear from
 \eqref{fpj} that $\sigma(\psi_j)/\psi_j$ is independent of $n$ and depends only on the parity of $j$. Since, for $n=1$, $P_1=1$ and $P_2=x$, it follows that 
\begin{equation*}\frac{\tilde\sigma(P_j)}{P_j}=\begin{cases}\tilde\sigma(1), & j \text{ odd},\\
\tilde\sigma(x)/x, & j \text{ even}.
\end{cases} \end{equation*}
Suppose now that $f$ is given by \eqref{put}, where $n=1$ and 
 $P$ vanishes at $\eta_j$ for some $j$. 
 Then,  $f(z+ 1/3)$ and  $f(z- 1/3)$ vanish at $z=\gamma_j$ and at $z=\gamma_j+1/3$, so $Q$ vanishes at $\xi_j$ and $\eta_j$. 
(Vanishing at $\eta_1=\infty$ should be interpreted as vanishing of the top coefficient.) 
Using this for $j=0$ and $j=1$ gives
\begin{align*}
\tilde\sigma(x)&=x(x-2\zeta-1)(ax+b),\\
\tilde\sigma(1)&=\big((\zeta+2)x-\zeta\big)(cx+d)
\end{align*}
for some constants $a$, $b$, $c$, $d$.
 For $j=2$ and $j=3$, we find that
$$(\zeta+2)\tilde\sigma(x)-(2\zeta+1)\tilde\sigma(1)
$$ 
vanishes at $x=\zeta(2\zeta+1)/(\zeta+2)$ and at  $x=(2\zeta+1)/(\zeta+2)$, and that
$$\tilde\sigma(x)-\zeta\tilde\sigma(1) $$
vanishes at $x=1$ and $x=\zeta$.  This gives a system of four linear equations for $a,b,c,d$.
The  solution space is one-dimensional and spanned by
$$a=\zeta+2,\qquad b=c=-3, \qquad d=2\zeta+1.$$
Thus, for some constant $C$, $\tilde\sigma(x)=Cxe_1(x)$ and $\tilde\sigma(1)=C e_0(x)$, so that $\tilde \sigma=C\hat \sigma$.

Finally, we note that if $f$ is given by  \eqref{tss}, then
$\sigma(f)/f$ is independent of $n$ and the parameters $b_j$.  
Since  \eqref{tss} is uniformized by 
\eqref{gpu}, it follows that
$$\frac{\hat\sigma\left((x-a)\prod_{k=1}^{n-1}(x-b_k)G(x,b_k)\right)}{(x-a)\prod_{k=1}^{n-1}(x-b_k)G(x,b_k)}=\frac{\hat\sigma(x-a)}{x-a}=\frac{xe_1(x)-ae_0(x)}{x-a}.$$
\end{proof}

Although it is not hard to give an explicit expression for the constant $C(\tau)$, the following somewhat implicit description will be more useful.

\begin{lemma}
Let 
$$\phi(z)=\frac{e^{-2\pi\ti z}\tha(e^{4\pi\ti z};p^2)}{C(\tau)\tha(\om pe^{\pm 2\pi\ti z};p^2)^{2}},$$
where $C$ is the constant in \emph{Lemma \ref{spl}}. Then, 
\begin{equation}\label{esu}
\sigma M =  \phi^{-1}M\hat\sigma,
\end{equation} where $M$ denotes multiplication by
\begin{equation}\label{m}M(z)=e^{-2\pi\ti z}\tha(e^{4\pi\ti z};p^2)\tha(\om pe^{\pm 2\pi\ti z};p^2)^{3n-2}.\end{equation}
Moreover,
\begin{equation}\label{phx}\phi(z)^2=-3X(x(z)),\end{equation}
 where
\begin{align}\nonumber X(x)&=(x-2\zeta-1)\big((\zeta+2)x-\zeta\big)\big((\zeta+2)x-\zeta(2\zeta+1)\big)(x-1)\\
\label{x} &= (\zeta+2)^2\prod_{j=0}^3(x-\xi_j).
\end{align}
\end{lemma}

\begin{proof}
The identity \eqref{esu} follows immediately from Lemma \ref{spl}.
To prove \eqref{phx}, we note that both sides are elliptic functions with the 
same zeroes and poles. Thus, it suffices to verify \eqref{phx} for a fixed $z$. 
To this end, we apply \eqref{esu} to 
$P_1$ and evaluate the result at $z=1/3$. Since $M(0)=0$, we get
$$\frac\ti{\sqrt 3}\,M(2/3)P_1(x(2/3))=\frac{M(1/3)}{\phi(1/3)}\,e_0(x(1/3))P_1(x(1/3)).$$
Since $x(1/3)=x(2/3)=0$, $P_1(0)\neq 0$ and $M(1/3)=-M(2/3)\neq 0$, we find that
$\phi(1/3)=\ti\sqrt 3\,e_0(0)=-\ti\sqrt 3\,\zeta(2\zeta+1)$. It follows that \eqref{phx} holds at
$z=1/3$, which completes the proof.
\end{proof}

The following result is a uniformization of the identity $\sigma^2=\id$.

\begin{lemma}\label{scl}
If  $\hat\sigma_n$ denotes the operator defined in
\emph{Lemma \ref{spl}}, then 
$$\hat\sigma_{n+2}Y\hat\sigma_n=-3XY,$$
 where $X$ is as in \eqref{x} and
\begin{equation}
\label{y}Y(x)=x\big((\zeta+2)x-(2\zeta+1)\big)(x-\zeta)=(\zeta+2)(x-\eta_0)(x-\eta_2)(x-\eta_3).
\end{equation}
\end{lemma}

\begin{proof}
Denoting the function \eqref{m} by  $M_n$ and using \eqref{esu}, we need to prove that
\begin{equation}\label{ssp}\phi M_{n+2}^{-1}\sigma M_{n+2}Y\phi M_n^{-1}\sigma M_n=-3XY. \end{equation}
Using  \eqref{xd}, it is easy to check that
$$  M_{n+2}Y\phi M_n^{-1}\sim e^{-6\pi\ti z}\tha(e^{12\pi\ti z};p^6). $$
This is a $1/3$-periodic function, and thus commutes with $\sigma$. Since $\sigma^2=\id$ as an operator on $\Theta_n$, \eqref{ssp} follows from \eqref{phx}.
\end{proof}

Finally, we mention the following easily verified identities. They 
uniformize a corresponding symmetry of the functions \eqref{fpj} under
$z\mapsto z+\tau/2$.

\begin{lemma}\label{pfl}
\begin{align*}x^{3n-2}P_j\left(\frac{\zeta(2\zeta+1)}{(\zeta+2)x}\right)&=
(-1)^{n+1}\zeta^{n-1+\left[\frac j2\right]}\left(\frac{2\zeta+1}{\zeta+2}\right)^{\left[\frac{3(j-1)}2\right]}{P_{2n+1-j}(x)},\\
x^{3n}(\hat\sigma P_j)\left(\frac{\zeta(2\zeta+1)}{(\zeta+2)x}\right)&=
(-1)^{n}\zeta^{n+\left[\frac j2\right]}\left(\frac{2\zeta+1}{\zeta+2}\right)^{\left[\frac{3(j-1)}2\right]+1}{(\hat\sigma P_{2n+1-j})(x)}
 \end{align*}
\end{lemma}

\subsection{Symmetric polynomials in two sets of variables}
Let us now consider the result of applying $\hat\sigma$ to some of the variables in
\eqref{tut}. For  $0\leq k\leq 2n$, let
\begin{multline}\label{ttv}T(x_1,\dots,x_k;x_{k+1},\dots,x_{2n})\\
=\frac{(\id^{\otimes k}\otimes\,\hat\sigma^{\otimes (2n-k)})\Delta(x_1,\dots,x_{2n})T(x_1,\dots,x_{2n})}{\Delta(x_1,\dots,x_k)\Delta(x_{k+1},\dots,x_{2n})}.
 \end{multline}
This is a polynomial, which is symmetric in its first $k$ and last $2n-k$ variables, and generalizes  $T(x_1,\dots,x_{2n})=T(x_1,\dots,x_{2n};-)$.
By construction,  the function
\begin{multline}\label{dps}\prod_{j=1}^k e^{-2\pi\ti z_j}\tha(e^{4\pi\ti z_j};p^2)\tha(\om pe^{\pm 2\pi\ti z_j};p^2)^{3n-2}\prod_{j=k+1}^{2n}\tha(\om pe^{\pm 2\pi\ti z_j};p^2)^{3n}\\
\times\Delta(x_1,\dots,x_k)\Delta(x_{k+1},\dots,x_{2n})
\,T(x_1,\dots,x_k;x_{k+1},\dots,x_{2n}), \end{multline}
where $x_j=x(z_j)$, spans the one-dimensional space
$(\id^{\otimes k}\otimes\,\sigma^{\otimes (2n-k)})\Theta_n^{\wedge 2n}$.

Applying Lemma \ref{spl} to \eqref{sdt} gives the determinant identity
\begin{equation}\label{sdd}T(x_1,\dots,x_k;x_{k+1},\dots,x_{2n})=\frac{\det_{1\leq i,j\leq 2n}(A_{ij})}{\Delta(x_1,\dots,x_k)\Delta(x_{k+1},\dots,x_{2n})}, \end{equation}
where
$$A_{ij}=\begin{cases}P_j(x_i), & 1\leq i\leq k,\\
(\hat\sigma P_j)(x_i), & k+1\leq i\leq 2n.
\end{cases} $$

One can also extend the alternative determinant formula \eqref{tn}.
In view of the ostensible asymmetry of the variables, 
we replace the operator 
$\id^{\otimes k}\otimes\,\hat\sigma^{\otimes (2n-k)}$ in \eqref{ttv} by
$$(-1)^{l(n-k)}\left(\id^{\otimes k}\otimes\,\hat\sigma^{\otimes(n-k)}\otimes
\id^{\otimes l}\otimes\,\hat\sigma^{\otimes(n-l)}\right).$$
The left-hand side is then replaced by 
$$T(x_1,\dots,x_k,x_{n+1},\dots,x_{n+l};x_{k+1},\dots,x_n,x_{n+l+1},\dots,x_{2n}). $$

Expanding the determinant in \eqref{tn} gives
$$\Delta(\mathbf x)T(\mathbf x)=\sum_{\sigma\in\mathrm{S}_n}\sgn(\sigma)\prod_{i=1}^n(x_i-x_{n+\sigma(i)})\prod_{i,j=1,\,j\neq \sigma(i)}^n(x_i-x_{n+j})G(x_i,x_{n+j}). $$
We can use Lemma \ref{spl} to compute the image of this expression when
$\hat\sigma$ acts on any subset of the variables.
Re-writing the result as a determinant gives
\begin{multline}\label{tkl}
T(x_1,\dots,x_k,x_{n+1},\dots,x_{n+l};x_{k+1},\dots,x_n,x_{n+l+1},\dots,x_{2n})\\
=\frac{\prod_{{1\leq i\leq k,\,l+1\leq j\leq n}}(x_{n+j}-x_i)\prod_{{k+1\leq i\leq n,\,1\leq j\leq l}}(x_i-x_{n+j})\prod_{i,j=1}^nG(x_i,x_{n+j})}{ \Delta(x_1,\dots,x_k)\Delta(x_{k+1},\dots,x_n)\Delta(x_{n+1},\dots,x_{n+l})\Delta(x_{n+l+1},\dots,x_{2n})}\\
\times\det(B),
\end{multline}
where 
$$B_{i,j}=\begin{cases}
\displaystyle\frac 1{G(x_i,x_{n+j})}, & 1\leq i\leq k,\ 1\leq j\leq l,\\[4mm]
\displaystyle\frac {Q(x_i,x_{n+j})}{(x_{n+j}-x_i)G(x_i,x_{n+j})}, & 1\leq i\leq k,\ l+1\leq j\leq n,\\[4mm]
\displaystyle\frac {Q(x_{n+j},x_i)}{(x_i-x_{n+j})G(x_i,x_{n+j})}, & k+1\leq i\leq n,\ 1\leq j\leq l,\\[4mm]
\displaystyle\frac {R(x_i,x_{n+j})}{G(x_i,x_{n+j})}, & k+1\leq i\leq n,\ l+1\leq j\leq n,
\end{cases} $$
with
\begin{align*}
Q(x,y)&=ye_1(y)-xe_0(y),\\
R(x,y)&=\frac{e_0(x)ye_1(y)-xe_1(x)e_0(y)}{y-x}\\
&=3(\zeta+2)^2x^2y^2+\zeta(\zeta+2)(2\zeta+1)(x^2+y^2)\\&\quad-2(\zeta^2+4\zeta+1)\big((\zeta+2)xy+\zeta(2\zeta+1)\big)(x+y)\\
&\quad+4(\zeta^4+4\zeta^3+8\zeta^2+4\zeta+1)xy
+3\zeta^2(2\zeta+1)^2.
\end{align*}

For fixed $k+l$, 
varying $k$ and $l$ in  \eqref{tkl} gives different expressions for the same quantity.
For instance,  to compute 
 $T(x_1,x_2;x_3,x_4)$  one may  use
\eqref{tkl} with $k=2$, $l=0$, which gives
\begin{multline*}T(x_1,x_2;x_3,x_4)=\big((x_4-x_1)(x_3-x_2)G(x_1,x_4)G(x_2,x_3)Q(x_1,x_3)Q(x_2,x_4)\\
-(x_3-x_1)(x_4-x_2)G(x_1,x_3)G(x_2,x_4)Q(x_1,x_4)Q(x_2,x_3)\big)/(x_2-x_1)(x_4-x_3)
\end{multline*}
or with $k=l=1$, which gives
\begin{multline*}T(x_1,x_2;x_3,x_4)=(x_4-x_1)(x_3-x_2)G(x_1,x_4)G(x_2,x_3)R(x_3,x_4)\\
-G(x_1,x_2)G(x_3,x_4)Q(x_1,x_4)Q(x_2,x_3).\end{multline*}

The following  elementary result will be useful.

\begin{lemma}\label{tstvl}
One has
\begin{multline}\label{tred}T(t,x_2,\dots,x_{k};x_{k+1},\dots,x_{2n-1},t)\\=(-1)^{k+1}2G(t,t)\prod_{j=2}^{2n-1}G(x_j,t)\,T(x_2,\dots,x_{k};x_{k+1},\dots,x_{2n-1}).
 \end{multline}
\end{lemma}

To see this, 
expand the determinant in  \eqref{tkl} along the first row and then let
$x_1=x_{2n}=t$. Then, only the term corresponding to the last column gives a non-zero contribution. Rewriting the complementary minor in terms of the polynomial $T$, using also $Q(t,t)=2G(t,t)$, gives \eqref{tred} after simplification.

Applying the Jacobi--Desnanot identity (see e.g.\ \cite[Prop.\ 10]{kr}) to \eqref{tkl} gives several different recursions. 

\begin{lemma}\label{jdl}
The polynomial $T$ satisfies 
\begin{subequations}\label{jd}
\begin{multline}\label{lc}(a-b)(c-d)T(\mathbf x;\mathbf y)T(a,b,c,d,\mathbf x;\mathbf y)
=G(a,d)G(b,c)T(a,c,\mathbf x;\mathbf y)T(b,d,\mathbf x;\mathbf y)\\
-G(a,c)G(b,d)T(a,d,\mathbf x;\mathbf y)T(b,c,\mathbf x;\mathbf y),
\end{multline}
\begin{multline}
(a-b)T(\mathbf x;\mathbf y)T(a,b,c,\mathbf x;d,\mathbf y)
=(a-d)G(a,d)G(b,c)T(a,c,\mathbf x;\mathbf y)T(b,\mathbf x;d,\mathbf y)\\
-(b-d)G(a,c)G(b,d)T(a,\mathbf x;d,\mathbf y)T(b,c,\mathbf x;\mathbf y),
\end{multline}
\begin{multline}\label{dlc}
T(\mathbf x;\mathbf y)T(a,c,\mathbf x;b,d,\mathbf y)
=(d-a)(b-c)G(a,d)G(b,c)T(a,c,\mathbf x;\mathbf y)T(\mathbf x;b,d,\mathbf y)\\
-
G(a,c)G(b,d)T(a,\mathbf x;d,\mathbf y)T(c,\mathbf x;b,\mathbf y).
\end{multline}
\begin{multline}\label{lcr}(a-b)(c-d)T(\mathbf x;\mathbf y)T(\mathbf x;a,b,c,d,\mathbf y)
=G(a,d)G(b,c)T(\mathbf x;a,c,\mathbf y)T(\mathbf x;b,d,\mathbf y)\\
-G(a,c)G(b,d)T(\mathbf x;a,d,\mathbf y)T(\mathbf x;b,c,\mathbf y),
\end{multline}
\end{subequations}
\end{lemma}

\subsection{The functions $T_n^{(\mathbf k)}$}\label{tnks}

As in \S \ref{sts}, let $n\in\mathbb Z$ and $\mathbf k\in\mathbb Z^4$,
with $m=2n-|\mathbf k|\geq 0$. Given this data, we
 will  define a certain function
$T_n^{(\mathbf k)}$ of $m$ variables. In the special case when $k_j\geq 0$ for each $j$, 
it is given by the specialization
$$T_n^{(\mathbf k)}(x_1,\dots,x_m)=T(x_1,\dots,x_m,\boldsymbol\xi^{\mathbf k}), $$
 where
$$\boldsymbol\xi^{\mathbf k}=\left(\xi_0^{(k_0)},\xi_1^{(k_1)},\xi_2^{(k_2)},\xi_3^{(k_3)}\right)\\
=(\underbrace{\xi_0,\dots,\xi_0}_{k_0},\underbrace{\xi_1,\dots,\xi_1}_{k_1},\underbrace{\xi_2,\dots,\xi_2}_{k_2},\underbrace{\xi_3,\dots,\xi_3}_{k_3}).$$
To motivate the extension to negative $k_j$, we note that for $\mathbf l\in\mathbb Z_{\geq 0}^4$, 
\begin{equation}\label{dit}T_n^{(\mathbf k)}(x_1,\dots,x_m)=\frac{(-1)^{\binom{|l|}2}\,T(x_1,\dots,x_m,\boldsymbol\xi^{{\mathbf k+\mathbf l}};\boldsymbol\xi^{{\mathbf l}})}{2^{|l|}\prod_{i,j=0}^3G(\xi_i,\xi_j)^{(k_i+l_i)l_j}\prod_{j=1}^m\prod_{i=0}^3G(x_j,\xi_i)^{l_i}}. \end{equation}
This is easily seen by iterating Lemma \ref{tstvl}. 
Relaxing the restriction that each $k_j$ is positive, we
 use \eqref{dit} as the definition of the left-hand side, where
$\mathbf l,\,\mathbf k+\mathbf l\in\mathbb Z_{\geq 0}^4$. 
The most economic expression is obtained for
 $l_j=\max(0,-k_j)=k_j^-$, giving
\begin{equation}\label{tnk}T_n^{(\mathbf k)}(x_1,\dots,x_{m})=\frac{(-1)^{\binom{|\mathbf k^-|}2}\,T(x_1,\dots,x_m,\boldsymbol\xi^{{\mathbf k}^+};\boldsymbol\xi^{{\mathbf k}^-})}{2^{|\mathbf k^-|}\prod_{i,j=0}^3G(\xi_i,\xi_j)^{k_i^-k_j^+}\prod_{j=1}^m\prod_{i=0}^3G(x_j,\xi_i)^{k_i^-}}.\end{equation}

To give an example,
\begin{align*}T_0^{(-2,1,0,0)}(x)&=-\frac{T(x,\xi_1;\xi_0,\xi_0)}{4G(\xi_0,\xi_1)^2G(x,\xi_0)^2},\\
&=\frac{(2\zeta+1)^2(\zeta+2)}{\zeta^2 x^3}\left((\zeta^2+\zeta+1)x(2\zeta+1-x)+\zeta(2\zeta+1)^2\right).
\end{align*}
In general, as a function of the variables $x_j$, 
 $T_n^{(\mathbf k)}$ is a symmetric rational function with poles
only at the points $x_j=\eta_i$, where $k_i<0$, $i\neq 1$.

The freedom to vary $\mathbf l$ in \eqref{dit}  implies the 
following fundamental property.

\begin{lemma}\label{tkrl}
For $\mathbf l\in\mathbb Z_{\geq 0}^4$,  $T_n^{(\mathbf k+\mathbf l)}(x_1,\dots,x_m)=T_n^{(\mathbf k)}(x_1,\dots,x_m,\xi^{\mathbf l})$.
\end{lemma}

The function $T_n^{(\mathbf k)}$ allows us to construct elements
of the space $(\Theta_n^{\mathbf k})^{\wedge m}$.

\begin{proposition}\label{tnkp}
The function
\begin{multline}\label{ukl}\prod_{j=1}^{m}\left({e^{-2\pi\ti z_j}\theta(e^{4\pi\ti z_j};p^2)\theta(\omega p e^{\pm 2\pi \ti z_j};p^2)^{3n-2}}{\prod_{l=0}^3(x_j-\xi_l)^{k_l}}\right)\\
\times \Delta(x_1,\dots,x_{m})T_{n}^{(\mathbf k)}(x_1,\dots,x_{m}), \end{multline}
 is an element of
 $(\Theta_n^{\mathbf k})^{\wedge m}$. 
\end{proposition}

We will prove later (Corollary \ref{tnvc}) that $T_{n}^{(\mathbf k)}$
does not vanish identically, and deduce as a consequence (Theorem \ref{dt})
that $\dim\Theta_n^{\mathbf k}=m$. It follows that \eqref{ukl}
in fact spans the space  $(\Theta_n^{\mathbf k})^{\wedge m}$.

To prove Proposition \ref{tnkp} we note that, by Lemma \ref{kl}, we can construct an element 
in  $(\Theta_n^{\mathbf k})^{\wedge m}$ by starting from 
$\Psi\in(\id^{\otimes(m+|\mathbf k^+|)}\otimes\,\sigma^{\otimes|\mathbf k^-|})
 \Theta_{n+|\mathbf k^-|}^{\wedge 2(n+|\mathbf k^-|)}
$, 
applying  specialized derivatives to the last $|\mathbf k^+|+|\mathbf k^-|$
variables and finally dividing by  $\Phi$.
For instance, if $\mathbf k=(2,-2,0,0)$ we should take
\begin{multline*}\frac{1}{\prod_{j=1}^m\tha(e^{6\pi\ti(\gamma_1\pm z_j)};p^6)^2}\,\frac{\partial^3}{\partial z_{m+1}^3}\frac{\partial}{\partial z_{m+2}}\frac{\partial^2}{\partial z_{m+3}^2}\frac{\partial^0}{\partial z_{m+4}^0}\Bigg|_{z_{m+1}=z_{m+2}=\gamma_0,\,z_{m+3}=z_{m+4}=\gamma_1}\\
\Psi(z_1,\dots,z_m,z_{m+1},z_{m+2},z_{m+3},z_{m+4}).\end{multline*}

Returning to the general case, we
choose $\Psi$ as in  \eqref{dps} and observe that
$\tha(e^{4\pi\ti z};p^2)$ has a single zero at $z=\gamma_j$ for each $j$. 
Moreover, by \eqref{xp},  $x(z)-x(w)$ has a double zero at $z=w$ when $w=\gamma_j$. Thus, all derivatives must hit the prefactor  and we are left with a specialization of the polynomial $T$. 
Finally, we uniformize $\Phi$ using \eqref{gu}. 
This leads to  an expression proportional to \eqref{ukl}.

Lemma \ref{tkrl} shows that the set of functions $T_n^{(\mathbf k)}$
is closed under specializing variables to $\xi_j$. There are similar results  for specializations to $\eta_j$ with $j\neq 1$. We have proved that, if $k_j\geq 0$, then
$T_n^{(\mathbf k)}(\mathbf x,\eta_j)$ is an elementary factor times $T_{n-1}^{(\mathbf k-\mathbf e_j)}(\mathbf x)$.  Moreover, if $k_j<0$, $\lim_{a\rightarrow\eta_j}T_n^{(\mathbf k)}(\mathbf x,a)/(a-\eta_j)^{2k_j+1}$ is an elementary factor times $T_n^{(\mathbf k+\mathbf e_j)}(\mathbf x)$. Modified versions of these statements hold also for $j=1$.

\subsection{Modularity}\label{ms}
We will now discuss the behaviour of  $x$ and $\zeta$ under modular transformations. 
We first recall some basic facts on modular functions, see e.g.\ \cite{sh}. Let $\Gamma$ be one of the congruence groups
\begin{align*}\Gamma_0(n)&=\left\{\left(\begin{matrix}a&b\\c&d\end{matrix}\right)\in \mathrm{SL}(2,\mathbb Z);\, c\equiv 0\ \operatorname{mod}\ n \right\},\\
\Gamma_0(m,n)&=\left\{\left(\begin{matrix}a&b\\c&d\end{matrix}\right)\in \mathrm{SL}(2,\mathbb Z);\, c\equiv 0\ \operatorname{mod}\ m,\  b\equiv 0\ \operatorname{mod}\ n\right\}.
\end{align*}
In fact, $\Gamma_0(m,n)\simeq \Gamma_0(mn)$ as the two groups are conjugate in $\mathrm{SL}(2,\mathbb R)$.
The group $\Gamma$ acts on the upper half-plane $\mathbb H$, and on  $\mathbb Q\cup\{\infty\}$, by
$$A.\tau=\frac{a\tau+b}{c\tau+d},\qquad A=\left(\begin{matrix}a&b\\c&d\end{matrix}\right). $$
We will write  $\mathcal F(\Gamma)$ for the corresponding field of modular functions, that is, meromorphic functions $f$ on  $\mathbb H$ such that  
$f(A.\tau)=f(\tau)$ for $A\in\Gamma$. 
The normalizer
$\mathrm N(\Gamma)=\mathrm N_{\mathrm{SL}(2,\mathbb R)}(\Gamma)/\Gamma$ 
acts naturally on $\mathbb H/\Gamma$ and  $\mathcal F(\Gamma)$.

If $\tau\in \mathbb Q\cup\{\infty\}$ is the unique 
fix-point of some element in $\Gamma$, then the
 $\Gamma$-orbit of $\tau$ is called a  \emph{cusp}.  
Denoting by $\overline{\mathbb H/\Gamma}$  the union of
$\mathbb H/\Gamma$ and the set of cusps, 
one can equip  $\overline{\mathbb H/\Gamma}$ with the structure of a compact Riemann surface, so that any modular function  extends meromorphically to the compactification.

From now on, we let $\Gamma=\Gamma_0(6,2)\simeq\Gamma_0(12)$. The corresponding equivalence of modular functions is
\begin{equation}\label{ggc}f\in\mathcal F(\Gamma_0(6,2))\quad \Longleftrightarrow\quad (\tau\mapsto f(2\tau))\in\mathcal F(\Gamma_0(12)).\end{equation}
It follows from  \cite[Thm.\ 8]{al}
that  $\mathrm N(\Gamma)\simeq\mathrm S_2\times \mathrm S_3$.
To describe  $\mathrm N(\Gamma)$ explicitly, we note that
$\Gamma_0(3)=\mathrm N_{\mathrm{SL}(2,\mathbb Z)}(\Gamma)\subseteq \mathrm N_{\mathrm{SL}(2,\mathbb R)}(\Gamma)$. The
projection $\Gamma_0(3)\rightarrow \mathrm N_{\mathrm{SL}(2,\mathbb R)}(\Gamma)/\Gamma$ amounts to reducing all matrix elements mod $2$, which gives a surjection to $\mathrm{SL}(2,\mathbb Z/2\mathbb Z)\simeq \mathrm S_3$. 
We choose  $\tau\mapsto \tau/(3\tau+1)$
and $\tau\mapsto(2\tau+1)/(3\tau+2)$
as representatives of two generators for  $\mathrm S_3$.
Moreover, $\mathrm N_{\mathrm{SL}(2,\mathbb R)}(\Gamma)$ contains the
 Fricke involution
$\tau\mapsto -1/3\tau$.
These three transformations indeed generate  $\mathrm S_2\times \mathrm S_3$.

 The elements in $\mathbb Q\cup\{\infty\}$ 
that are the unique fix-point of some element in $\Gamma$ 
are $\infty$ and the  numbers $k/6l$ with $(k,l)=1$. 
The cusps are given by 
$$C_0,\ C_2,\ C_3,\ C_4,\ C_6,\ C_{\infty}, $$
where 
\begin{align*}C_j&=\left\{\frac k{6l};\ (k,l)=1,\ k\equiv\pm j\ \operatorname{mod}\ 12\right\},\qquad j\neq\infty,\\
C_\infty&=\left\{\frac k{6l};\ (k,l)=1,\ k\equiv\pm 1\ \operatorname{mod}\ 6\right\}\cup\{\infty\}. \end{align*}
The action of $\mathrm S_2\times\mathrm S_3$ extends to a permutation action on the cusps. Explicitly,  the generator $\tau\mapsto \tau/(3\tau+1)$ corresponds to 
$C_2\leftrightarrow C_\infty$, $C_3\leftrightarrow C_6$,  the generator $\tau\mapsto (2\tau+1)/(3\tau+2)$ to $C_0\leftrightarrow C_3$, $C_4\leftrightarrow C_\infty$ and the generator $\tau\mapsto -1/3\tau$  to
$C_0\leftrightarrow C_\infty$, $C_2\leftrightarrow C_6$, $C_3\leftrightarrow C_4$,
Note that the $S_3$-part of the action preserves the two sets of cusps 
$\{C_2,C_4,C_\infty\}$ and $\{C_0,C_3,C_6\}$. We will refer to the cusps in the first set as \emph{trigonometric} and the others as \emph{hyperbolic}. 
As  motivation for this terminology, note that
if $\tau$ is purely imaginary, the corresponding Jacobi theta functions 
degenerate to trigonometric functions as $\tau\rightarrow \ti\infty\in C_\infty$
and to hyperbolic functions as $\tau\rightarrow 0\in C_0$.
As we will see,  $T_n^{(\mathbf k)}$ behaves quite differently at these two types of cusps.

We will now consider  the 
 modular properties of
the function $x$. We first observe that, if $f(z)= f(z,\tau)$ is an even  elliptic function with periods $1$ and $\tau$, then the same is true for the function
$$\hat f(z)= f\left(\frac z{c\tau+d},\frac{a\tau+b}{c\tau+d}\right), \qquad 
\left(\begin{matrix}a&b\\c&d\end{matrix}\right)\in \mathrm{SL}(2,\mathbb Z).
 $$ 
We consider the case $f(z)=x(z+1/2,\tau)$. 
Up to quasi-periodicity and evenness, the  zeroes and poles 
 of  $\hat f$ are a single zero at $z_0=(c\tau+d)/6$ 
and a single pole at $z_\infty=(a\tau+b)/2+(c\tau+d)/6$.
 Since $f(0)=\hat f(0)=1$, it follows that 
\begin{equation}\label{fmt}\hat f(z)=\frac{(1-f(z_\infty))(f(z)-f(z_0))}{(1-f(z_0))(f(z)-f(z_\infty))}. \end{equation}
In particular,  if
 $A=\left(\begin{smallmatrix}a&b\\c&d\end{smallmatrix}\right)\in\Gamma_0(3)$, then $\hat f$ depends only on the residue class of  $A$ in 
 $\mathrm{SL}(2,\mathbb Z/2\mathbb Z)$. This leads to symmetries of $x$ 
under  $S_3$. Using  
 Lemma~\ref{xvl}, we can  rewrite \eqref{fmt} for two generators of $\mathrm S_3$ as
\begin{subequations}\label{xmt}
\begin{align}
x\left(\frac{z-1/2}{3\tau+1}+\frac 12,\frac{\tau}{3\tau+1}\right)&=\frac 1{x},\\
x\left(\frac{z-1/2}{3\tau+2}+\frac 12,\frac{2\tau+1}{3\tau+2}\right)&=\frac{(\zeta+2)x-(2\zeta+1)}{1-\zeta},
\end{align}
\end{subequations}
where $x=x(z,\tau) $ and $\zeta=\zeta(\tau)$.

Specializing $z=1/2+1/3$ in \eqref{xmt}, it follows that 
$\zeta\in\mathcal F(\Gamma_0(6,2))$, and that $\zeta$ transforms under the $\mathrm{S}_3$ action by 
\begin{subequations}\label{zsa}
\begin{align}\zeta\left(\frac{\tau}{3\tau+1}\right)&=\frac 1{\zeta}, \\
\zeta\left(\frac{2\tau+1}{3\tau+2}\right)&=-\zeta-1,
\end{align}
\end{subequations}
where $\zeta=\zeta(\tau)$.
Moreover, for the Fricke involution,
$$\zeta\left(-\frac 1{3\tau}\right)=\frac{1-\zeta}{1+2\zeta}. $$
One way to check this is to use
 \cite[Lemma 9.1]{r} to express both sides in terms of Dedekind's eta function.
Since it can also be
deduced from \cite[Table 3]{m} (see below) we do not give the details.
Using \eqref{zsa} and Lemma \ref{hpl}, 
 \eqref{xmt} can be written more compactly as
\begin{align*}
x\left(\frac{z}{3\tau+1},\frac{\tau}{3\tau+1}\right)&=\frac {(\zeta+2)x}{\zeta(2\zeta+1)},\\
x\left(\frac{z}{3\tau+2},\frac{2\tau+1}{3\tau+2}\right)&=\frac{(\zeta+1)x}{\zeta-x}.
\end{align*}

The meromorphic extension of $\zeta$ to the cusps is given by 
\begin{align*}\zeta(C_2)&=-\frac 12,& \zeta(C_4)&=1,& \zeta(C_\infty)&=-2,\\
\zeta(C_0)&=-1,& \zeta(C_3)&=0,& \zeta(C_6)&=\infty.\end{align*}
Indeed, it follows from \eqref{z} that
\begin{equation}\label{ztl}\zeta=-2+6p+\mathcal O(p^2),\qquad \tau\rightarrow\ti\infty, \end{equation}
which covers the case of $C_\infty$. The 
 other values now follow using the symmetries;
for instance,
$$\zeta(C_2)=\zeta\left(\frac{C_\infty}{3 C_\infty+1}\right)=\frac 1{\zeta(C_\infty)}=-\frac 12. $$

Another consequence of \eqref{ztl} is that   $\zeta+2$ has a simple zero at $C_\infty$. Applying the symmetries, $\zeta$ has a simple pole at $C_6$. Since $\zeta$ has no  poles in $\mathbb H$ or at the other cusps, it follows from
 \cite[Thm.\ 4.24]{fo} that $\zeta$ is  a bijection
from
 $\overline{\mathbb H/\Gamma_0(6,2)}$ to $\mathbb C\cup\{\infty\}$.
Thus, $\zeta$ generates the field $\mathcal F(\Gamma)$, that is, it is a \emph{Hauptmodul}.
In particular, we recover the known fact that 
$\overline{\mathbb H/\Gamma_0(6,2)}$  is  a sphere.

By \eqref{ggc}, the function $\tau\mapsto\zeta(2\tau)$ is a Hauptmodul for
$\Gamma_0(12)$. In the literature, one usually considers Hauptmoduln that vanish at $C_\infty$ and have a pole  at $C_0$. Let $t=t_{12}$ be such a Hauptmodul, normalized as in \cite{m}. Comparing the values at cusps (see \cite[Table 2]{m}), we find that 
\begin{equation}\label{ztr}\zeta(2\tau)=-\frac{t+4}{t+2}.\end{equation}
This is also easy to see by comparing explicit formulas. 
By routine manipulation of infinite products,
$$
\zeta(2\tau)=-2\frac{\eta(\tau)^3 \eta(4\tau)^6\eta(6\tau)^3}{\eta(2\tau)^9\eta(3\tau)\eta(12\tau)^2},
$$
where 
$\eta(\tau)=p^{1/12}(p^2;p^2)_\infty$
is the Dedekind eta function.
One can then verify  \eqref{ztr} from
\cite[Table 3]{m}.

Finally, we comment  on the role of the cusps in relation to
 the three-colour model and the XYZ spin chain.
In the three-colour  model, the states are three-colourings of a square lattice, where adjacent squares have distinct colour, and the  colours have independent weights $t_0$, $t_1$, $t_2$. The relations between these weights and the variable $\zeta$ can be written \cite{r}
$$\frac{(t_0t_1+t_0t_2+t_1t_2)^3}{(t_0t_1t_2)^2}=27+\frac{(\zeta-1)^4(\zeta+2)(2\zeta+1)}{\zeta(\zeta+1)^4}. $$
In particular, the trigonometric cusps $\zeta=-1/2$, $\zeta=1$ and $\zeta=-2$  correspond to a parameter regime containing the enumeration point $t_0=t_1=t_2$ when all states have equal weight. The hyperbolic cusps 
$\zeta=-1$, $\zeta=0$ and $\zeta=\infty$ correspond to the limit cases $t_j\rightarrow\ 0$ and $t_j\rightarrow\infty$, and thus  control the maximal and minimal number of squares of each colour; cf.\ \cite[Cor.~3.3]{r}. 

In the conventions of \cite{zj}, the coupling constants of the XYZ chain are
$$J_x=\frac 1{1-\zeta},\qquad J_y=\frac 1{1+\zeta}, \qquad J_z=-\frac 12 $$
(here we have used \eqref{zzr}). Note that
$J_xJ_y+J_xJ_z+J_yJ_z=0$, which is equivalent to the supersymmetry 
condition $\Delta=-1/2$. We see that three the trigonometric cusps correspond to the conditions
$$J_x=J_y,\qquad J_x=J_z,\qquad J_y=J_z $$
and thus to the XXZ chain, whereas the  hyperbolic cusps correspond to
$$J_x=\infty,\qquad J_y=\infty,\qquad J_z=\infty, $$ 
that is, to the XY chain.

\subsection{Symmetries}

The polynomials $T$ have a useful symmetry under the group $\mathrm S_4$, acting by  rational transformations on the variables $x_j$ and $\zeta$. This is easy to understand geometrically. As we observed in \S \ref{sts}, if $f$ is in $\Theta_n$ then so is $f(z+\gamma_j)$
for $j=0,3$ and  $e^{6\pi\ti n z}f(z+\gamma_j)$ for $j=1,3$.
These maps induce an action of $\mathrm S_2\times \mathrm S_2$ on the space $\Theta_n^{\wedge 2n}$ (which comes from a genuine action on   $\Theta_n$ when $n$ is even but a projective action when $n$ is odd). By Lemma \ref{hpl}, 
this  results in an $\mathrm S_2\times \mathrm S_2$ symmetry of $T$ under simultaneous rational transformations of the variables $x_j$. Moreover, it is easy to check that 
if $z\mapsto f(z,\tau)$ is in $\Theta_n$, then so is
$$z\mapsto \exp\left({-\frac{6\pi \ti ncz^2}{c\tau +d}}\right)f\left(\frac{z}{c\tau+d},\frac{a\tau +b}{c\tau +d}\right),\qquad \left(\begin{matrix}a&b\\c&d\end{matrix}\right)\in \Gamma_0(3), $$
the condition $c\equiv 0\ \operatorname{mod}\ 3$ being necessary for 
preserving property \eqref{fe}.
It follows that, if  $\Psi(z_1,\dots,z_{2n};\tau)\in\Theta_n^{\wedge 2n}$, then
$$
\Psi(z_1,\dots,z_{2n};\tau)\sim
\exp\left(-\frac{6\pi \ti n c\sum_{k=1}^{2n}z_k^2}{c\tau +d}\right)
\Psi\left(\frac{z_1}{c\tau+d},\dots,\frac{z_{2n}}{c\tau+d};\frac{a\tau +b}{c\tau +d}\right).$$
By \eqref{xmt} and \eqref{zsa},  these modular transformations
induce an $S_3$ symmetry of the polynomials $T$.
 Taken together, translations by half-periods and  modular transformations 
generate an action of $\mathrm S_4$  described in Proposition \ref{msl}.

We stress that there is no symmetry of the polynomials $T$
corresponding to the Fricke involution. Accordingly,
the behaviour of $T$ at the trigonometric  and the hyperbolic cusps
 is quite different, see \S \ref{tls} and \S \ref{ntcs}.

The above considerations only determine the  $\mathrm S_4$ symmetry
up to  factors depending on $\zeta$. 
To derive it explicitly, we note that (writing $G(x,y;\zeta)=G(x,y)$ etc.)
\begin{align}\label{gs}
\notag G(x,y;\zeta)&=\zeta^2 x^2y^2G\left(\frac 1x,\frac 1y;\frac 1\zeta\right)\\
\notag&=\left(\frac{\zeta-1}{\zeta+2}\right)^2G\left(\frac{(\zeta+2)x-(1+2\zeta)}{1-\zeta},\frac{(\zeta+2)y-(1+2\zeta)}{1-\zeta};-\zeta-1\right)\\
&=\left(\frac{(\zeta+2)}{\zeta(2\zeta+1)}\right)^2x^2y^2G\left(\frac{\zeta(2\zeta+1)}{(\zeta+2)x},\frac{\zeta(2\zeta+1)}{(\zeta+2)y};\zeta\right),\\
\notag Q(x,y;\zeta)&=\zeta^2 xy^3Q\left(\frac 1x,\frac 1y;\frac 1\zeta\right)\\
\notag &=\left(\frac{\zeta-1}{\zeta+2}\right)^2Q\left(\frac{(\zeta+2)x-(1+2\zeta)}{1-\zeta},\frac{(\zeta+2)y-(1+2\zeta)}{1-\zeta};-\zeta-1\right)\\
\notag &=\left(\frac{(\zeta+2)}{\zeta(2\zeta+1)}\right)^2xy^3Q\left(\frac{\zeta(2\zeta+1)}{(\zeta+2)x},\frac{\zeta(2\zeta+1)}{(\zeta+2)y};\zeta\right),\\
\notag R(x,y;\zeta)&=\zeta^4 x^2y^2R\left(\frac 1x,\frac 1y;\frac 1\zeta\right)\\
\notag &=\left(\frac{\zeta-1}{\zeta+2}\right)^2R\left(\frac{(\zeta+2)x-(1+2\zeta)}{1-\zeta},\frac{(\zeta+2)y-(1+2\zeta)}{1-\zeta};-\zeta-1\right)\\
\notag &=\left(\frac{(\zeta+2)}{\zeta(2\zeta+1)}\right)^2x^2y^2R\left(\frac{\zeta(2\zeta+1)}{(\zeta+2)x},\frac{\zeta(2\zeta+1)}{(\zeta+2)y};\zeta\right).
\end{align}
Using these identities in \eqref{tkl} yields the following symmetries. 

\begin{proposition}\label{msl}
Writing $T(\mathbf x;\mathbf y;\zeta)=T(\mathbf x;\mathbf y)$,
the polynomial $T$ satisfies
\begin{subequations}
\begin{align}
\nonumber\lefteqn{T(x_1,\dots,x_{k};x_{k+1},\dots,x_{2n};\zeta)}\\
\label{sia}&=\zeta^{2(n(n+1)-k)}\prod_{j=1}^{k}x_j^{3n-k-1}\prod_{j=k+1}^{2n}x_j^{n+k+1}\, T(x_1^{-1},\dots,x_{k}^{-1};x_{k+1}^{-1},\dots,x_{2n}^{-1};\zeta^{-1})\\
&\nonumber=\left(\frac{\zeta-1}{\zeta+2}\right)^{n(n+1)+2nk-k(k+1)}\\
\label{sib}&\quad\times T\left(\frac{(\zeta+2)x_1-(1+2\zeta)}{1-\zeta},\dots;\dots,\frac{(\zeta+2)x_{2n}-(1+2\zeta)}{1-\zeta};-\zeta-1\right)\\
\notag&=\left(\frac{\zeta+2}{\zeta(2\zeta+1)}\right)^{n(n+1)+2nk-k(k+1)}
\prod_{j=1}^{k}x_j^{3n-k-1}\prod_{j=k+1}^{2n}x_j^{n+k+1}\\
\label{sic}&\quad\times T\left(\frac{\zeta(2\zeta+1)}{(\zeta+2)x_1},\dots;\dots,\frac{\zeta(2\zeta+1)}{(\zeta+2)x_{2n }};\zeta\right).
\end{align}
\end{subequations}
\end{proposition}

Composing these three 
 involutions indeed generate  the  group $\mathrm S_4$. 
Using  \eqref{tnk} and \eqref{gs} gives 
 corresponding symmetries of the polynomials  $T_n^{(\mathbf k)}$, with
$\mathrm S_4$ permuting the  indices $k_j$.

 \begin{corollary}\label{sscc}
The polynomials  $T_n^{(\mathbf k)}$ satisfy
\begin{align*}
\notag\lefteqn{T_n^{(k_0,k_1,k_2,k_3)}(x_1,\dots,x_{m};\zeta)}\\
\notag&=\zeta^{2n(n-1)}\prod_{j=0}^3\xi_j^{k_j(n-1)}\prod_{j=1}^{m}x_j^{n-1}\, T_n^{(k_1,k_0,k_2,k_3)}(x_1^{-1},\dots,x_{m}^{-1};\zeta^{-1})\\
&\notag=\left(\frac{\zeta-1}{\zeta+2}\right)^{n(n-1)}\\
&\notag\quad\times T_n^{(k_2,k_1,k_0,k_3)}\left(\frac{(\zeta+2)x_1-(1+2\zeta)}{1-\zeta},\dots,\frac{(\zeta+2)x_{m}-(1+2\zeta)}{1-\zeta};-\zeta-1\right)\\
&\notag=\left(\frac{\zeta+2}{\zeta(2\zeta+1)}\right)^{n(n-1)}
\prod_{j=0}^3\xi_j^{k_j(n-1)}
\prod_{j=1}^{m}{x_j^{n-1}}\\
&\quad\times T_n^{(k_1,k_0,k_3,k_2)}\left(\frac{\zeta(2\zeta+1)}{(\zeta+2)x_1},\dots,\frac{\zeta(2\zeta+1)}{(\zeta+2)x_{m }};\zeta\right).
\end{align*}
\end{corollary}

Next, we describe a different 
 symmetry  for  the polynomials  $T$, which interchanges the left and right group of variables.
The  reason is that multiplication by $e^{-6\pi\ti z}\tha(e^{12\pi\ti z};p^6)$ defines an isomorphism from $\Theta_n$ to the subspace of $\sigma\Theta_{n+2}$ consisting of functions vanishing at the lattice \eqref{lal}. 
Thus, the alternants of the two spaces are proportional. 
Applying $\sigma$ to some  of the variables and uniformizing leads to the following identity, 
except for the identification of the prefactor.

\begin{proposition}\label{tlrs}
One has
\begin{multline*}T(x_1,\dots,x_k;x_{k+1},\dots,x_{2n},\xi_0,\xi_1,\xi_2,\xi_3)\\
=-\frac{3^{n-k+1}(2\zeta(\zeta+1))^{4n+6}(\zeta-1)^2(2\zeta+1)^2}{(\zeta+2)^{2n+2k+4}}\prod_{j=1}^{2n}Y(x_j)\\
\times T(x_{k+1},\dots,x_{2n};x_1,\dots,x_k),
 \end{multline*}
with $Y$ is as in \eqref{y}.
\end{proposition}

\begin{proof}
We first prove the case  $k=2n$, that is,
\begin{equation}\label{lzs}T(x_1,\dots,x_{2n};\xi_0,\xi_1,\xi_2,\xi_3)
=C_n\prod_{j=1}^{2n}Y(x_j)T(-;x_{1},\dots,x_{2n}),\end{equation}
where 
\begin{equation}\label{tsc}C_n=-\frac{(2\zeta(\zeta+1))^{4n+6}(\zeta-1)^2(2\zeta+1)^2}{3^{n-1}(\zeta+2)^{6n+4}}. \end{equation}
We prove this by induction on $n$, using the recursions of Lemma \ref{jdl}.
Having verified \eqref{lzs} for $n=0$ and $n=1$, 
we start from \eqref{lc} with $\mathbf x=(x_1,\dots,x_{2n})$, $\mathbf y=(\xi_0,\xi_1,\xi_2,\xi_3)$. Applying \eqref{lzs}, 
the right-hand side of \eqref{lc} becomes $C_{n+1}^2$ times the right-hand side of \eqref{lcr}, with $\mathbf y$ replaced by $\mathbf x$ and $\mathbf x$ replaced by $\emptyset$. Since $T(-;\mathbf x)$ does not vanish identically, this proves \eqref{lzs}, where $C_{n+2}C_n=C_{n+1}^2$, which is consistent with \eqref{tsc}.

We now multiply \eqref{lzs} by $\Delta(\mathbf x)$ and apply $\hat\sigma$ to the variables $x_{k+1},\dots,x_n$. Writing $\mathbf x'=(x_1,\dots,x_k)$, $\mathbf x''=(x_{k+1},\dots,x_{2n})$, the  left-hand side becomes
$$(\zeta+2)^{-2(2n-k)}\Delta(\mathbf x')\Delta(\mathbf x'')\prod_{i=k+1}^{2n}X(x_i)\,T(\mathbf x';\mathbf x'',\xi_0,\xi_1,\xi_2,\xi_3) $$
and the right-hand side, using also Lemma \ref{scl},
\begin{multline*}(-1)^kC_n\Delta(\mathbf x')\prod_{i=1}^kY(x_i)\prod_{i=k+1}^{2n}(\hat\sigma_{x_i}Y(x_i)\hat\sigma_{x_i})\Delta(\mathbf x'')T(\mathbf x'';\mathbf x')\\ 
=C_n3^{2n-k}\Delta(\mathbf x')\Delta(\mathbf x'')\prod_{i=1}^{2n}Y(x_i)\prod_{i=k+1}^nX(x_i)T(\mathbf x'';\mathbf x').
\end{multline*}
 This completes the proof.
 \end{proof}

Proposition \ref{tlrs} yields an additional involution symmetry of
the polynomials $T_n^{(\mathbf k)}$, but only for $m=0$.

\begin{corollary}\label{zscc}
For $m=2n-|\mathbf k|=0$ and $\boldsymbol \delta=(1,1,1,1)$,
\begin{multline*}T_n^{(\mathbf k)}=\frac{(-1)^{n+1}(\zeta+2)^{2(k_1+k_2+n+2)}}{12^{n+1}\zeta^{2(k_1+k_2+2n+3)}(\zeta-1)^{2(k_2+k_3+1)}(\zeta+1)^{2(k_0+k_1+2n+3)}(2\zeta+1)^{2(k_0+k_2+1)}}\\
\times T_{-n-2}^{(-\mathbf k-\boldsymbol \delta)}. \end{multline*}
\end{corollary}

\begin{proof}
We start from \eqref{dit} with the additional assumption $l_j\geq 1$.  We can then apply Proposition \ref{tlrs}, and conclude that
\begin{multline*}T_n^{(\mathbf k)}=\frac{(-1)^{n+1}(\zeta-1)^2(2\zeta+1)^2(\zeta+2)^{2n}}{12^{n+1}\big(\zeta(\zeta+1)\big)^{2(2n+1)}}\\
\times\prod_{j=0}^3\left(\frac 1{Y(\xi_j)^{k_j+1}}\left(\frac{\left(2\zeta(\zeta+1)\right)^4Y(\xi_j)^2}{(\zeta+2)^4\prod_{i=0}^3G(\xi_i,\xi_j)}\right)^{k_j+l_j}\right)T_{-n-2}^{(-\mathbf k-\boldsymbol\delta)}. \end{multline*}
It follows in particular that  the factor raised to $k_j+l_j$  is identically $1$. Inserting the explicit factorization of $Y(\xi_j)$, we arrive at the desired result.
\end{proof}

In \cite{r2}, we will show that the case $m=0$ of the polynomials
$T_n^{(\mathbf k)}$ can be identified with tau functions of Painlev\'e VI. This identification implies one further symmetry, which we state here for completeness. It would be interesting to obtain a more direct proof.
Let
$(Y_k)_{k\in\mathbb Z}$ be the solution to the recursion
\begin{equation}\label{yrec}Y_{k+1}Y_{k-1}=2(2k+1)Y_k^2,\qquad Y_0=Y_1=1,
\end{equation}
that is,
$$Y_k=\begin{cases}
\prod_{j=1}^k \frac{(2j-1)!}{(j-1)!}, & k\geq 0,\\[1mm]
\frac{(-1)^{\frac{k(k+1)}{2}}}{2^{2k+1}}\prod_{j=1}^{-k-1} \frac{(2j-1)!}{(j-1)!}, & k<0.
\end{cases} $$

\begin{proposition}[\cite{r2}, Cor.\ 5.3]\label{esp}
For $m=0$,
\begin{multline}
\notag T_n^{(k_0,k_1,k_2,k_3)}=(-1)^{(k_0+k_1+n)(k_1+k_3+n)}\frac{Y_{n-k_0}Y_{n-k_1}Y_{n-k_2}Y_{n-k_3}}{Y_{k_0}Y_{k_1}Y_{k_2}Y_{k_3}}\\
\times\left(\frac{\zeta^{k_1+k_2-n}(\zeta+1)^{k_0+k_1-n}}{(\zeta-1)^{k_0+k_1-n}(\zeta+2)^{k_1+k_2-n}(2\zeta+1)^{k_1+k_3-n}}\right)^{n-1} T_n^{(n-k_0,n-k_1,n-k_2,n-k_3)}.
\label{tse}\end{multline}
\end{proposition}

By Corollary \ref{zscc} and Proposition \ref{esp},  
 the  $\mathrm S_4$ symmetry of Corollary \ref{sscc} is enhanced to an  $\mathrm S_4\times  \mathrm S_2\times \mathrm S_2$ symmetry when $m=0$.

\section{The trigonometric cusp $\zeta=-2$}
\label{tls}

In this Section,  we will consider the behaviour of $T_n^{(\mathbf k)}$ 
at the trigonometric cusp $\zeta= -2$
 (or, equivalently, $\tau=\ti\infty$, $p=0$).
A careful investigation of this limit will lead to a proof of
 Theorem \ref{dt}. The behaviour at the other two trigonometric cusps
can be deduced using the symmetries of Corollary \ref{sscc}.

\subsection{Polynomials related to symplectic and orthogonal characters}

For $n$ a non-negative integer,  let $V_n$ be  the space of Laurent polynomials $f(t)$ satisfying 
\begin{equation}\label{vnd}f(t^{-1})=-f(t), \qquad f(t)+f(\om t)+f(\om^2 t)=0, 
 \end{equation}
such that $t^{n+[(n-1)/2]}f(t)$  is regular at $t=0$. 
A natural basis for $V_n$ consists of the polynomials 
$t^{-m}-t^m$, where $1\leq m\leq n+[(n-1)/2]$ and $m\not\equiv 0\ \operatorname{mod}\ 3$ or, equivalently, $m=j+[(j-1)/2]$,  $1\leq j\leq n$. In particular, $\dim(V_n)=n$. Note that, when $p=0$, \eqref{tss} 
is an element of $V_{2n}$, where the variables are related by $t=e^{2\pi\ti z}$.
Thus,  $V_{2n}$ appears as the trigonometric limit of $\Theta_{n}$.

In terms of the variable $t$, the map \eqref{sig} takes the form
$$(\sigma f)(t)=\frac{\ti}{\sqrt 3}\left(f(\om t)-f(\om^2 t)\right), $$
that is,
\begin{equation}\label{sitr}\sigma(t^{m})=\begin{cases}
-t^{m}, & m\equiv 1 \ \operatorname{mod}\ 3,\\
t^{m}, & m\equiv 2 \ \operatorname{mod}\ 3.
\end{cases} \end{equation}
It is a bijection from $V_n$ to the space of Laurent polynomials satisfying
$$f(t^{-1})=f(t), \qquad f(t)+f(\om t)+f(\om^2 t)=0, $$
such that $t^{n+[(n-1)/2]}f(t)$  is regular at $t=0$.

The one-dimensional space $(\id^k \otimes\,\sigma^{n-k})V_n$
is spanned by $\det_{1\leq i,j\leq n}(X_{ij})$, where
\begin{equation}\label{xm}X_{ij}=\begin{cases} t_i^{-(j+[(j-1)/2])}-t_i^{j+[(j-1)/2]}, & 1\leq i\leq k,\\
t_i^{-j-[(j-1)/2]}+t_i^{j+[(j-1)/2]}, & k+1\leq i\leq n,\ j \text{ odd},\\
-t_i^{-j-[(j-1)/2]}-t_i^{j+[(j-1)/2]}, & k+1\leq i\leq n,\ j \text{ even}.
\end{cases} \end{equation}
Since the determinant vanishes if $t_i^2=1$ for $1\leq i\leq k$
or if $t_i=t_j^{\pm}$, where $1\leq i<j\leq k$ or  $k+1\leq i<j\leq n$,
it is  natural to define 
\begin{multline}\label{chi}\chi(t_1,\dots,t_k;t_{k+1},\dots,t_n)\\
=\frac{\det_{1\leq i,j\leq n}(X_{ij})}{\prod_{i=1}^kt_i^{-k}(1-t_i^2)\prod_{i=k+1}^nt_i^{1-n+k}\prod_{1\leq i<j\leq k \text{ or } k+1\leq i<j\leq n}(t_i-t_j)(1-t_it_j)},
 \end{multline}
which is a polynomial in the variables $t_j+t_j^{-1}$, and symmetric in its first $k$ and its last $n-k$ variables. 

Recall \cite{fuh} that the characters of the symplectic and even orthogonal groups are given by
\begin{align*}\chi_\lambda^{\mathfrak{sp}(2n)}(t_1,\dots,t_n)&=\frac{\det_{1\leq i,j\leq n}
\left(t_i^{-(\lambda_j+n-j+1)}-t_i^{\lambda_j+n-j+1}\right)}{\prod_{i=1}^nt_i^{-n}
(1-t_i^2)\prod_{1\leq i<j\leq n}(t_j-t_i)(1-t_it_j)},\\
\chi_\lambda^{\mathfrak{o}(2n)}(t_1,\dots,t_n)&=\frac{\det_{1\leq i,j\leq n}
\left(t_i^{-(\lambda_j+n-j)}+t_i^{\lambda_j+n-j}\right)}{C\prod_{i=1}^nt_i^{1-n}
\prod_{1\leq i<j\leq n}(t_j-t_i)(1-t_it_j)},\end{align*}
where $\lambda=(\lambda_1,\dots,\lambda_n)$ is a partition, $C=2$ if $\lambda_n= 0$ and $C=1$ else.
It is  easy to check that
\begin{subequations}\label{wc}
\begin{align}\label{ssc}\chi(t_1,\dots,t_n;-)&=\chi_{\left[\frac{n-1}2\right],\left[\frac{n-2}2\right],\dots,1,1,0,0}^{\mathfrak{sp}(2n)}(t_1,\dots,t_n), \\
\label{soc}\chi(-;t_1,\dots,t_n)&=(-1)^{[n/2]}\chi_{\left[\frac{n+1}2\right],\left[\frac{n}2\right],\dots,2,2,1,1}^{\mathfrak{o}(2n)}(t_1,\dots,t_n). \end{align}
\end{subequations}
Indeed, compared to 
\eqref{chi}, the order of the columns is reversed, which is compensated for by the change of sign in the factor $t_i-t_j$ in the denominator.
In the case of \eqref{soc}, we also multiply the even rows by $-1$, and the
number of such rows is $[n/2]$.
It follows from \eqref{wc} that
\begin{equation}\label{wdp}\chi\left(1^{(n)};-\right)>0, \qquad (-1)^{[n/2]}\chi\left(-;1^{(n)}\right)>0.\end{equation}
Indeed, these numbers are the dimension of the corresponding representation space, given explicitly by Weyl's dimension formula.

We state some useful properties of the polynomials $\chi$.

\begin{lemma}\label{xil}
The polynomials $\chi$ satisfy
$$\chi(-t_1,\dots,-t_k;-t_{k+1},\dots,-t_n)=\pm\chi(t_1,\dots,t_k;t_{k+1},\dots,t_n), $$
where the sign is negative if and only if $n\equiv 2k+1\ \operatorname{mod}\ 4$.
\end{lemma}

\begin{proof}
Replacing all the variables  by their negative, the $j$:th column in \eqref{xm} changes sign if $j+[(j-1)/2]$ is odd. The number of such columns is odd unless $n\equiv 0\mod 4$. Moreover, the denominator in \eqref{chi} is multiplied by
$(-1)^{k+\binom{k}2+\binom{n-k}2}$, which is positive if and only if 
 $n\equiv 0$ or $n\equiv 2k+1\mod 4$. 
\end{proof}

If $f$  satisfies \eqref{vnd}, then  $(\sigma f)(-1)=(2\ti/\sqrt 3)f(-\om)$.
This implies the following result.

\begin{lemma}\label{xssl}
One has
\begin{multline*}\chi(t_1,\dots,t_k;t_{k+1},\dots,t_n,-1)\\
=-\frac{2\prod_{i=1}^k(1-t_i-t_i^{-1})}{\prod_{i=k+1}^nt_i^{-1}(1+t_i)^2}\,\chi(t_1,\dots,t_k,-\om;t_{k+1},\dots,t_n).\end{multline*}
\end{lemma}

The following closely related identities can be viewed as  trigonometric limits of Proposition~\ref{tlrs}. 

\begin{lemma}
The polynomials $\chi$ satisfy
\begin{multline}\label{xsi}\chi(t_1,\dots,t_k;t_{k+1},\dots,t_{n},1,-1)\\
=(-1)^{[n/2]+k(n-k)}2\prod_{j=1}^{n}\left(t_j^2+t_j^{-2}+1\right)\chi(t_{k+1},\dots,t_{n};t_1,\dots,t_k),\end{multline}
\begin{multline}\label{osia}\chi(t_1,\dots,t_k;t_{k+1},\dots,t_{n},1)\\
=(-1)^{[n/2]+k(n-k)+1}\prod_{j=1}^{n}\frac{t_j+t_j^{-1}-1}{t_j+t_j^{-1}+1}\,\chi(t_{k+1},\dots,t_{n};t_1,\dots,t_k,-1).\end{multline}
\end{lemma}

\begin{proof}
We first prove the case $k=n$ of \eqref{xsi}, that is,
\begin{equation}\label{osi}\chi(-;t_1,\dots,t_n,1,-1)=(-1)^{[n/2]}2\prod_{j=1}^n\left(t_j^2+t_j^{-2}+1\right)\chi(t_1,\dots,t_n;-). \end{equation}
Clearing denominators, this identity takes the form
\begin{equation}\label{osc}\det_{1\leq i,j\leq n+2}\left(X_{ij}^{(1)}\right)=(-1)^{[n/2]+1}8\prod_{j=1}^n(t_j^{-3}-t_j^3)\det_{1\leq i,j\leq n}\left(X_{ij}^{(2)}\right), \end{equation}
where $X^{(1)}$ and $X^{(2)}$ are the special cases of \eqref{xm} corresponding to the left-hand and right-hand sides of \eqref{osi}.
We note that the Laurent polynomial $\prod_{j=1}^n(t_j^{-3}-t_j^3)$ divides the left-hand side of \eqref{osc}. Indeed, if $t_j$ is a sixth root of unity, then the $j$:th row of $X^{(1)}$ is equal to one of the last two rows.
It is then straight-forward to check that the quotient is in the one-dimensional space $V_n^{\wedge n}$, so that \eqref{osc} holds up to a multiplicative constant. To identify the constant, we compute the coefficient
of $\prod_{j=1}^nt_j^{-(j+[(j-1)/2]+3)}$ for the two sides.
On the right-hand side, only the diagonal entries in $X^{(2)}$ contribute, so the coefficient is $(-1)^{[n/2]+1}8$. On the left-hand side, we get a contribution from the diagonal in the upper left $n\times n$-block, and from the complementary $2\times 2$-minor. This yields 
 a factor $(-1)$ from the entry $X_{jj}^{(1)}$ when $1\leq j\leq n$ and $j$ even, contributing in total $(-1)^{[n/2]}$. Moreover, $\det_{n+1\leq i,j\leq n+2}(X_{ij}^{(1)})=-8$. This completes the proof of \eqref{osi}. 
Applying $\sigma$ to the last $n-k$ variables in \eqref{osc} yields \eqref{xsi}. The proof of \eqref{osia} is similar and we do not give the details.
\end{proof}

\subsection{Factorization in the trigonometric limit}

By the following result,   appropriate specializations of
the polynomials $T$ factor in the trigonometric limit
as a product of generalized characters.
 Note that the relation $x_j=-(1+t_j+t_j^{-1})$ between the variables is natural since 
 $x(z)=-(1+e^{2\pi\ti z}+e^{-2\pi\ti z})$ when $p=0$.

\begin{proposition}\label{tdtp}
Let $k$, $l$, $m$, $n$ and $N$  be non-negative integers 
with $0\leq k\leq m$, $0\leq l\leq n$ and $m+n=2N$. Then, 
\begin{multline}\label{tcl}\lim_{\zeta\rightarrow -2}\left(\frac{\zeta+2}6\right)^{n(2N+k)-l(2k+n)+(2l-n)(N-1)-\delta(2l-n-1) }\\
\begin{split}&\quad\times T\left(x_1,\dots,x_k,\frac{\zeta(2\zeta+1)}{(\zeta+2)y_1},\dots,\frac{\zeta(2\zeta+1)}{(\zeta+2)y_l}; \right.\\
&\quad\qquad\qquad\qquad\left.x_{k+1},\dots,x_m,\frac{\zeta(2\zeta+1)}{(\zeta+2)y_{l+1}},\dots,\frac{\zeta(2\zeta+1)}{(\zeta+2)y_n}\right)\\
&=(-1)^{kl+k+l+m} 2^{N(N-1)}6^{m+n-k-l}\\
& \quad\times\prod_{j=1}^l\frac 1{(1+u_j+u_j^{-1})^{N+m+n-k-l-1}}\prod_{j=l+1}^n\frac 1{(1+u_j+u_j^{-1})^{N+k+l+1}}\\
&\quad\times
\chi(t_1,\dots,t_k;t_{k+1},\dots,t_{m})\,\chi(u_1,\dots,u_l;u_{l+1},\dots,u_n)
,\end{split}\end{multline}
where $x_j=-(1+t_j+t_j^{-1})$, $y_j=-(1+u_j+u_j^{-1})$ and $\delta(N)=[N^2/4]$.
\end{proposition}

\begin{proof}
Starting from \eqref{sdd}, we 
permute the rows so that the $i$:th row involves the variable $x_i$
for $1\leq i\leq m$ and $y_{i-m}$ for $m+1\leq i\leq m+n$.  
To the last $n$ rows, we apply Lemma \ref{pfl}
(with $n$ replaced by $N$).
 Let $\rho_j=[3(j-1)/2]$ and $\xi=\zeta(2\zeta+1)/(\zeta+2)$.  We pull out the factor
 $\xi^{\rho_m}$ from the $i$:th row when  $m+1\leq i\leq m+l$,
 $\xi^{\rho_{m}+1}$ from the $i$:th row when  $m+l+1\leq i\leq m+n$ and 
  $\xi^{\rho_j-\rho_m}$ from the $j$:th column, when $m+1\leq j\leq m+n$. 
At the point $\zeta=-2$, the resulting matrix is regular with
the upper right $n\times m$ block vanishing. Thus, the  determinant
 factors as the product of the diagonal blocks.

To describe these blocks, we introduce the notation
$$T^{\trig}(x_1,\dots,x_k;x_{k+1},\dots,x_{n})=\frac{\det_{1\leq i,j\leq n}(Y_{ij})}{\Delta(x_1,\dots,x_k)\Delta(x_{k+1},\dots,x_{n})},$$
where
\begin{equation}\label{yd}Y_{ij}=\begin{cases}P_j^\trig(x_i), & 1\leq i\leq k,\\
e_1^\trig(x_i)P_j^\trig(x_i), & k+1\leq i\leq n,\  j \text{ even},\\
e_0^\trig(x_i)P_j^\trig(x_i), & k+1\leq i\leq n,\  j \text{ odd},
\end{cases} \end{equation}
with
$$P_j^{\trig}(x)=2^{[(j+1)/2]-N}P_j(x)\Big|_{\zeta=-2}=x^{j-1}(x+3)^{\left[\frac{j-1} 2\right]}, $$
$$e_0^{\trig}(x)=\frac 16\,e_0(x)\Big|_{\zeta=-2}=-x-1,\qquad e_1^{\trig}(x)=\frac 16\,e_1(x)\Big|_{\zeta=-2}=x+3.   $$
After straight-forward simplification,  
we  find that the left-hand side of \eqref{tcl} equals
\begin{multline}\label{tlf}(-1)^{kl+n(k+l+N)} 2^{N(N-1)}6^{m+n-k-l}\prod_{j=1}^l{(-y_j)^{k+l+1-m-n-N}}\prod_{j=l+1}^n{(-y_j)^{-k-l-1-N}}\\
\times T^{\trig}(x_1,\dots,x_k;x_{k+1},\dots,x_m)\,T^{\trig}(y_1,\dots,y_l;y_{l+1},\dots,y_n).
\end{multline}

We now observe that, with $x=-1-t-t^{-1}$,
\begin{equation}\label{ptr}(t^{-1}-t)P_j^{\trig}(x)=\begin{cases}(t^{-1}-t)(2-t^3-t^ {-3})^{\frac{j-1}2}, & j \text{ odd},\\
(t^{2}-t^{-2}+t-t^{-1})(2-t^3-t^ {-3})^{\frac{j-2}2}, & j \text{ even},
\end{cases} \end{equation}
which is an element of  $V_j\setminus V_{j-1}$. 
Thus, for some constant $C_n$,
\begin{equation}\label{ssd}\det_{1\leq i,j\leq n}\left((t_i^{-1}-t_i)P_j^{\trig}(x_i)\right)=C_n\,{\det_{1\leq i,j\leq n}
\left(t_i^{-j-[(j-1)/2]}-t_i^{j+[(j-1)/2]}\right)}.\end{equation}
Identifying  the coefficient of
$\prod_{i=1}^n t_i^{-i-[(i-1)/2]}$ gives
$C_n=\prod_{i=1}^n(-1)^{i-1+[(i-1)/2]}$, that is,
 $C_n=-1$ when $n\equiv 2\ \operatorname{mod}\ 4$ and $C_n=1$ else. 

Next, we observe that
\begin{align*}
P_j^\trig(x)e_0^\trig(x)&=(t+t^{-1})(2-t^3-t^{-3})^{\frac{j-1}2}, & j&  \text{ odd},\\
P_j^\trig(x)e_1^\trig(x)&=(t^2+t^{-2}-t-t^{-1})(2-t^3-t^{-3})^{\frac{j-2}2}, & j&  \text{ even},
\end{align*}
which, by \eqref{sitr} and \eqref{ptr},  equals
  $\sigma((t^{-1}-t)P_j^{\trig}(x))$.
Thus, applying $\sigma$ to the last $n-k$ rows in \eqref{ssd} gives
$$\prod_{i=1}^k(t_i^{-1}-t_i)\det(Y)=C_n\det(X), $$
where $X$ and $Y$ are as in \eqref{xm} and \eqref{yd}.
Using also
$$\Delta(x_1,\dots,x_k)=(-1)^{\binom k2}\prod_{i=1}^kt_i^{1-k}\prod_{1\leq i<j\leq k}(t_i-t_j)(1-t_it_j), $$
it follows that
$$T^{\trig}(x_1,\dots,x_k;x_{k+1},\dots,x_n)=C_n(-1)^{\binom k2+\binom{n-k}2}\chi(t_1,\dots,t_k;t_{k+1},\dots,t_n). $$
Inserting this in \eqref{tlf} and simplifying completes the proof.
\end{proof}

Applying Proposition \ref{tdtp} to \eqref{tnk} leads after  simplification to
\begin{multline}\label{tcls}\lim_{\zeta\rightarrow -2}\left(\frac{\zeta+2}6\right)^{(k_1+k_2)(n-1)-\delta(k_1+k_2-1) }T_n^{(\mathbf k)}(x_1,\dots,x_m)\\
=(-1)^{\binom{|\mathbf k^-|}2+|\mathbf k^-|(k_1^++k_2^+)+k_0^-+k_3^-+k_2(n-1)}2^{n(n-1)-|\mathbf k^-|}3^{|\mathbf k^-|-k_1(n-1)}\prod_{j=1}^m\frac{1}{x_j^{2k_0^-}(x_j+2)^{2k_3^-}}\\
\begin{split}&\quad\times
\chi\left(t_1,\dots,t_m,1^{(k_0^+)},(-1)^{(k_3^+)};1^{(k_0^-)},(-1)^{(k_3^-)}\right)\\
&\quad\times \chi\left(1^{(k_1^+)},(-1)^{(k_2^+)};1^{(k_1^-)},(-1)^{(k_2^-)}\right).
\end{split}\end{multline}

\subsection{Non-vanishing character values}
\label{nvcvs}

The key fact for proving Theorem \ref{dt} is that $T_n^{(\mathbf k)}$ does not vanish identically. By \eqref{tcls} and the fact that $k_j^+k_j^-=0$, this would follow from the non-vanishing of the quantities
\begin{subequations}
\begin{align}\label{xsa}&\chi\left(1^{(k)},(-1)^{(l)};-\right),\\
\label{xsb} &\chi\left(1^{(k)};(-1)^{(l)}\right),\\
\label{xsc}&\chi\left((-1)^{(k)};1^{(l)}\right),\\
\label{xsd} &\chi\left(-;1^{(k)},(-1)^{(l)}\right). 
\end{align}
\end{subequations}
By Lemma \ref{xil}, \eqref{xsc} can be reduced to \eqref{xsb}. 
Moreover, by  \eqref{xsi}, if $k,l\geq 1$ \eqref{xsd} can be reduced to \eqref{xsa}. The non-vanishing of \eqref{xsd} when $l=0$  follows from \eqref{wdp}, and the case $k=0$ then  follows using Lemma \ref{xil}. In conclusion, the non-vanishing of $T_n^{(\mathbf k)}$ would follow from the non-vanishing of \eqref{xsa} and \eqref{xsb}. 

To investigate these two cases, we will use the following trigonometric limits of the recursions \eqref{jd}. 

\begin{lemma}
The polynomial $\chi$ satisfies
\begin{subequations}
\begin{multline}\label{xrb}
\frac{(a-b)(1-ab)}{ab}\,\chi(\mathbf t;\mathbf u)\chi(\mathbf t,a,b,c;\mathbf u)\\
=g(b,c)\chi(\mathbf t,a,c;\mathbf u)\chi(\mathbf t,b;\mathbf u)-g(a,c)\chi(\mathbf t,a;\mathbf u)\chi(\mathbf t,b,c;\mathbf u),
\end{multline}
\begin{multline}\label{xrd}
\chi(\mathbf t;\mathbf u)\chi(\mathbf t,a,c;\mathbf u,b)
=\frac{(c-b)(1-bc)}{bc}\,g(b,c)
\chi(\mathbf t,a,c;\mathbf u)\chi(\mathbf t;\mathbf u,b)\\
-g(a,c)\chi(\mathbf t,a;\mathbf u)\chi(\mathbf t,c;\mathbf u,b).
\end{multline}
\end{subequations}
Here, $\mathbf t$ and $\mathbf u$ are  (possibly empty) vectors, $a$, $b$, $c$, $d$  are scalars and 
$$g(t,u)=\left(tu+1+\frac 1{tu}\right)\left(\frac tu+1+\frac ut\right).$$
\end{lemma}

\begin{proof}
To prove \eqref{xrb}, we let 
$$\mathbf x=\left(x_1,\dots,x_k,\frac{\zeta(2\zeta+1)}{(\zeta+2)w_1},\dots,\frac{\zeta(2\zeta+1)}{(\zeta+2)w_l}\right)$$
and $\mathbf y=(y_1,\dots,y_m)$ in \eqref{lc}, where $k+l+m=2N$.
 We  replace
 $d$ by $\zeta(2\zeta+1)/(\zeta +2)d$,
multiply both sides with $(\zeta+2)^{(2l+1)(N+m)+2-\delta(l-1)-\delta(l)}$
 and let $\zeta\rightarrow -2$. Using Proposition \ref{tdtp}
and
$$\lim_{\zeta\rightarrow -2}G\left(-1-t-t^{-1},-1-u-u^{-1}\right)=2g(t,u),$$
$$\lim_{\zeta\rightarrow -2}\left(\frac{\zeta+2}6\right)^2G\left(x,\frac{\zeta(2\zeta+1)}{(\zeta+2)y}\right)=\frac{2}{y^2}, $$
we find that all  factors involving the variables $w_j$  and $d$
cancel. After simplification and a  change of variables   we obtain \eqref{xrb}.
The identity \eqref{xrd} is proved in the same way starting from \eqref{dlc}.
\end{proof}


For the next result, we use the notation \cite{r}
\begin{align}\label{ph}\phi_n(x)&={}_2F_1\left(\begin{matrix}-n/2,\,-(n-1)/2\\n+3/2\end{matrix};x\right), \\
\notag\psi_n(x)&={}_3F_2\left(\begin{matrix}-n/2,\,-(n+1)/2,\,(n+5)/4\\n+3/2,\,(n+1)/4\end{matrix};x\right). \end{align}
Note that $\phi_n$ and $\psi_n$ are polynomials in $x$ with positive coefficients.

\begin{proposition}
When all but one variable is specialized to $1$, the symplectic character \eqref{ssc} is given by
\begin{multline}\label{shr}\chi(t,1^{(n)};-)\\
=\begin{cases}\displaystyle\frac{3^{n(n+2)/4}}{4^n}\frac{\prod_{j=1}^{n/2}(6j+2)(6j-1)!(2j)!}{\prod_{j=1}^{n}(2j+1)!}\left(\frac{(1+t)^2}t\right)^{\frac n2}\psi_n\left(y\right), & n \text{ even},\\[5mm]
\displaystyle\frac{3^{(n-1)(n+1)/4}}{4^{n-1}}\frac{\prod_{j=1}^{(n-1)/2}(6j+4)!(2j)!}{\prod_{j=1}^{n-1}(2j+3)!}\left(\frac{(1+t)^2}t\right)^{\frac {n-1}2}\phi_n\left(y\right), & n \text{ odd},
\end{cases} \end{multline}
where $y=-(t-1)^2/3(t+1)^2$.
\end{proposition}

\begin{proof}
In \eqref{xrb}, we let $\mathbf t=1^{(n)}$,  $\mathbf u=\emptyset$, $a=t$ and $b=c=1$. This gives 
\begin{multline}\label{scr}
-\frac{(t-1)^2}{t}\,\chi(1^{(n)};-)\chi(t,1^{(n+2)};-)\\
=9\chi(1^{(n+1)};-)\chi(t,1^{(n+1)};-)-
\left(1+t+t^{-1}\right)^2\chi(1^{(n+2)};-)\chi(t,1^{(n)};-).
\end{multline}
We need to prove that, with the initial conditions
$\chi(-;-)=\chi(t;-)=\chi(t,1;-)=1$, the recursion \eqref{scr} is solved by \eqref{shr}. 

Plugging \eqref{shr} into \eqref{scr}, 
using $\phi_n(0)=\psi_n(0)=1$, gives after simplification for even $n$
$$\frac{3(3n+5)(3n+8)}{(2n+3)(2n+5)}\,y\,\psi_{n+2}(y)=(3y+1)^2\phi_{n+1}(y)-(y-1)^2\psi_n(y) $$
and for odd $n$
\begin{equation}\label{phr}\frac{3(3n+4)(3n+7)}{(2n+3)(2n+5)}\,y\,\phi_{n+2}(y)=(3y+1)\psi_{n+1}(y)-(y-1)^2\phi_n(y). \end{equation}
These identities (which are trivial to verify and hold regardless of the parity of $n$) were given 
in  \cite[Eq.\ (8.19--8.20)]{r}.
\end{proof}

We are now ready to prove the non-vanishing of \eqref{xsa}.

\begin{lemma}\label{nvcl}
We have
$$\varepsilon(k,l)\chi\left(1^{(k)},(-1)^{(l)};-\right)>0, $$
where $\varepsilon(k,l)=(-1)^{[l/2]+1}$ if $l> k+1$ and $l\equiv k+2\ \operatorname{mod}\ 4$, and $\varepsilon(k,l)=(-1)^{[l/2]}$ else.
\end{lemma}

\begin{proof}
Let  $X(k,l)=\chi(1^{(k)},(-1)^{(l)};-)$. By Lemma \ref{xil},
$X(k,l)=\pm X(l,k)$, where the sign is negative if and only if 
 $k+l\equiv 3 \ \operatorname{mod}\ 4$, which is equivalent to 
 $\varepsilon(k,l)\varepsilon(l,k)=-1$. Thus, we may assume that $k\geq l$. 
The case $l=0$ follows from  \eqref{wdp}. When $l=1$, we must prove that
$\chi(1^{(k)},-1;-)>0$. If we let $t\rightarrow -1$ in \eqref{shr}, then only the highest coefficient of the polynomials $\phi_n$ and $\psi_n$ 
contribute to the limit. Since this coefficient is positive, the result follows.

Assuming from now on that   $k\geq l\geq 2$, we  choose
  $\mathbf t=(1^{(k-1)},(-1)^{(l-2)})$, $\mathbf u=\emptyset$,  $a=1$ and $b=c=-1$ in \eqref{xrb}. This gives
\begin{equation}\label{xrbs}4X(k-1,l-2)X(k,l)=X(k,l-2)X(k-1,l)-9X(k,l-1)X(k-1,l-1). \end{equation}
For each of the six occurrences $X(k',l')$  of $X$ 
 in \eqref{xrbs}, we have $l'\leq k'+1$, so  
$\varepsilon(k',l')=(-1)^{[l'/2]}$. 
It follows that $\varepsilon(k,l-2)\varepsilon(k-1,l)=\varepsilon(k-1,l-2)\varepsilon(k,l)=-1$, $\varepsilon(k,l-1)\varepsilon(k-1,l-1)=1$. Thus, the numbers $Y(k,l)=\varepsilon(k,l)X(k,l)$ satisfy
$$4Y(k-1,l-2)Y(k,l)=Y(k,l-2)Y(k-1,l)+9Y(k,l-1)Y(k-1,l-1),$$
which  implies $Y(k,l)>0$ by induction on $k+l$. 
\end{proof}

The recursion \eqref{xrbs} can in fact be solved explicitly, leading
to the identities
\begin{multline*}\chi_{(n,n-1,n-1,\dots,1,1,0,0)}^{\mathfrak{sp}(4n+2)}(1^{(2n+1-l)},(-1)^{(l)})
=(-1)^{\left[\frac l2\right]}3^{\frac 12((l-n)(l-n-1)+n(n-1))}\\
\times
\prod_{j=1}^{n+1}\frac{(3j-2)!}{(n+j)!}
\prod_{j=1}^{\max(n-l,l-n-1)}\frac{(j-1)!(6j-1)!(n+j)!(2n-2j+1)!}{(2j-1)!(3j-1)!(n-j)!(2n+2j+1)!},
\end{multline*}
\begin{multline*}\chi_{(n-1,n-1,\dots,1,1,0,0)}^{\mathfrak{sp}(4n)}(1^{(2n-l)},(-1)^{(l)})
=(-1)^{\left[\frac {l+1}2\right]+\min(l,n)}2^{2n+|l-n|}3^{\frac{|l-n|(|l-n|+1)+n(n-5)}{2}}\\
\times
\frac{(3/2)_n}{(4/3)_n}
\prod_{j=1}^{n+1}\frac{(3j-2)!}{(n+j)!}\prod_{j=1}^{|l-n|}\frac{(j-1)!(6j-5)!(n+j-1)!(2n-2j+1)!}{(2j-2)!(3j-3)!(n-j)!(2n+2j-1)!}.
\end{multline*}
Although we have no need for these evaluations, we include  them here since they
may have some independent interest.

Next, we consider the quantity \eqref{xsb}.

\begin{lemma}\label{nvdl}
We have
$$\varepsilon(k,l)\,\chi\left(1^{(k)};(-1)^{(l)}\right)>0, $$
where $\varepsilon(k,l)=(-1)^{[(l+1)/2]+kl+1}$ if
 $l>k+2$ and $l\equiv k+3 \ \operatorname{mod}\ 4$, and $\varepsilon(k,l)=(-1)^{[(l+1)/2]+kl}$ else.
\end{lemma}

\begin{proof} The proof is similar to that of Lemma \ref{nvcl}. Let $X(k,l)=\chi(1^{(k)};(-1)^{(l)})$.
Combining Lemma \ref{xil} and \eqref{osia}, we deduce that 
$X(k,l+1)=\pm 3^{k-l}X(l,k+1)$, where the sign is negative if and only if $k+l\equiv 1\ \operatorname{mod}\ 4$, which is equivalent to
$\varepsilon(k,l+1)\varepsilon(l,k+1)=-1$. Thus, we may assume that $k+1\geq l$. The case $l=0$ follows from  \eqref{wdp}. When $l=1$, we use 
Lemma \ref{xssl} to write  $\chi(1^{(k)};-1)=(-1)^{k+1}2\chi(1^k,-\om;-)$. 
We then apply \eqref{shr}, using that, for $t=-\om$, $(1+t)^2/t=3$ and 
$-(t-1)^2/3(t+1)^2=1/9$. Since  $\phi_n $ and 
 $\psi_n$ have positive coefficients, it follows that $(-1)^{k+1}\chi(1^{(k)};-1)>0$, which is the case $l=1$. 

In remaining cases,  $k+1\geq l\geq 2$, we can use \eqref{xrd} to write
$$X(k,l)X(k-1,l-2)=4X(k,l-2)X(k-1,l)+9X(k,l-1)X(k-1,l-1). $$
If $X(k',l')$ is any occurrence of $X$ in this identity, then 
$k'+2\geq l'$. 
 It follows that
$\varepsilon(k,l-2)\varepsilon(k-1,l)=\varepsilon(k-1,l-2)\varepsilon(k,l)=\varepsilon(k,l-1)\varepsilon(k-1,l-1)=(-1)^{l+1}$. Thus, the numbers $Y(k,l)=\varepsilon(k,l)X(k,l)$ satisfy
$$Y(k-1,l-2)Y(k,l)=4Y(k,l-2)Y(k-1,l)+9Y(k,l-1)Y(k-1,l-1),$$
which gives $Y(k,l)>0$ by induction on $k+l$.
\end{proof}

As was explained at the beginning of \S \ref{nvcvs}, Lemma \ref{nvcl} and Lemma \ref{nvdl}
have the following important consequence.

\begin{corollary}\label{tnvc}
The left-hand side of \eqref{tcls} and, consequently,
the function $T_n^{(\mathbf k)}$, do not vanish identically.
\end{corollary}

We can now prove Theorem \ref{dt}. By general linear algebra, if  
$L_1,\dots ,L_M$ are  linear functionals on  a vector space $V$ of dimension $N\geq M$, then
\begin{multline}\label{dkt}\dim\left(\bigcap_{j=1}^M\Ker(L_j)\right)=N-M\\
\Longleftrightarrow\quad
\left(\id^{\otimes (N-M)}\otimes L_1\otimes\dots\otimes L_M\right)V^{\wedge N}\neq \{0\}.
\end{multline}
We apply this to the situation described in Lemma \ref{kl}, where 
$\Theta_n^{\mathbf k}$ is obtained from $\Theta_{n+|\mathbf k^-|}$ by imposing a number of linear conditions. By construction (see the proof of
Proposition \ref{tnkp}), the space corresponding to the
right side of \eqref{dkt} is spanned by \eqref{ukl}. Thus,  Theorem \ref{dt}
follows from Corollary \ref{tnvc}.

\section{The hyperbolic cusp $\zeta=0$}
\label{ntcs}

In this section we study the behaviour of $T_n^{(\mathbf k)}$ at $\zeta=0$
(the behaviour at the other two hyperbolic cusps 
follows using Corollary \ref{sscc}).
This might at first seem  quite easy since, as $\zeta\rightarrow 0$,
$$G(x,y)=2xy(x+y-1)+\mathcal O(\zeta),$$
and thus \eqref{tn} reduces to the Cauchy determinant.  It follows that
$$T(x_1,\dots,x_{2N};-)\bigg|_{\zeta=0}=2^{N(N-1)}x_1^{N-1}\dotsm x_{2N}^{N-1}, 
$$
which is more elementary than the symplectic character appearing for $\zeta=-2$.
However, the situation is  more complicated when
some variables are specialized to 
$\xi_1$ or $\xi_2$, which tend to $0$ as $\zeta\rightarrow 0$.

\subsection{Statement of results}
\label{hsrs}

The  behaviour of $T_n^{(\mathbf k)}$ at $\zeta=0$ is   different
in different   regimes for the
 discrete parameters. Recall that $n,k_0,k_1,k_2,k_3\in\mathbb Z$ 
are such that  $m=2n-\sum_j k_j\in\mathbb Z_{\geq 0}$. We say that these parameters belong to regime A, B and C if, respectively,
\begin{align}\tag{A}  |k_1+k_2+1|&\leq m+|k_0+k_3+1|,\\
\tag{B} k_1+k_2+1&\leq -(m+|k_0+k_3+1|)
,\\
\tag{C} k_1+k_2+1&\geq m+|k_0+k_3+1| .
\end{align}
For fixed parameters, we define
$$ L=\begin{cases}(k_1+k_2)(2n-k_1-k_2-1) & \text{in regime A},\\
 (k_1+k_2)(2n-k_1-k_2-1)+(n+1)(n-k_0-k_3) & \text{in regime B},\\
(k_1+k_2)(2n-k_1-k_2-1)+(n+1-m)(k_1+k_2-n)& \text{in regime C}.
\end{cases}
$$
This is consistent at the overlaps of the regimes, as can be seen
 from 
\begin{align*}(n+1)(n-k_0-k_3)&=\frac{(k_1+k_2+1+m)^2-(k_0+k_3+1)^2}4,\\
(n+1-m)(k_1+k_2-n)&=\frac{(k_1+k_2+1-m)^2-(k_0+k_3+1)^2}4.
\end{align*}

We will denote by
$J(N;a,b,c,d)$ the subset of $\{0,1,\dots,N\}$ consisting of the $a$ smallest and  $b$ largest odd elements as well as the $c$ smallest and $d$ largest even elements.
We can then formulate the main result of \S \ref{ntcs} as follows.
By $\sim$, we mean that equality holds up to a  factor independent  
 of the variables $x_j$.
It is clear from the proof that this factor
can be computed explicitly, but for the sake of simplicity we have not done so.

\begin{theorem}\label{ht}
The function $T_n^{(\mathbf k)}$ is exactly $\mathcal O(\zeta^L)$ as $\zeta\rightarrow 0$, that is,  $C_n^{(\mathbf k)}=\lim_{\zeta\rightarrow 0}T_n^{(\mathbf k)}/\zeta^L$ exists and is not identically zero. 
When the parameters are in regime A and $k_0+k_3+1\geq 0$, let
$N=2(n+k_1^-+k_2^-)$. Then,
\begin{subequations}\label{ca}
\begin{multline}C_n^{(\mathbf k)}\sim\frac 1{\Delta(\mathbf x)}\prod_{j=1}^m\frac{x_j^{n-k_1-k_2-1}}{(1-x_j)^{k_0+k_3+1}(1-2x_j)^{2k_2^-}}\sum_{\lambda_1,\dots,\lambda_m=0}^{N}\Delta(\boldsymbol \lambda)\\
\times\prod_{j=1}^m\left(\binom{N}{\lambda_j}(1-2x_j)^{\lambda_j}\prod_{\mu\in J(N;k_2^+,k_1^+,k_2^-,k_1^-)}(\lambda_j-\mu)\right).\end{multline}
When the parameters are in regime A and $k_0+k_3+1\leq 0$, let
$N=2(n+k_1^-+k_2^--k_0-k_3-1)$. Then,
\begin{multline} C_n^{(\mathbf k)}\sim\frac 1{\Delta(\mathbf x)}\prod_{j=1}^m\frac{x_j^{n+k_0+k_3-k_1-k_2}}{(1-2x_j)^{2k_2^-}}\sum_{\lambda_1,\dots,\lambda_m=0}^{N}\Delta(\boldsymbol \lambda)\\
\times\prod_{j=1}^m\left(\binom{N}{\lambda_j}(1-2x_j)^{\lambda_j}(1+2(-1)^{\lambda_j})\prod_{\mu\in J(N;k_2^+,k_1^+,k_2^-,k_1^-)}(\lambda_j-\mu)\right).\end{multline}
\end{subequations}
When the parameters are in regime B,
$$C_n^{(\mathbf k)}\sim\prod_{j=1}^m\frac{x_j^{k_0+k_3-k_1-k_2-1}(1-2x_j)^{2k_2+1}}{(1-x_j)^{n+1}}\prod_{1\leq i<j\leq m}(x_i+x_j-1). $$
Finally, when the parameters are in regime C,
$$C_n^{(\mathbf k)}\sim\prod_{j=1}^m{x_j^{k_0+k_3}(1-x_j)^{k_1+k_2-n}}\prod_{1\leq i<j\leq m}(x_i+x_j-1). $$
\end{theorem}

As will be clear from the proof, there are also determinant formulas  for $C_n^{(\mathbf k)}$ that may seem more transparent than \eqref{ca}. One advantage of the expressions \eqref{ca} is that they are easily seen 
not to vanish identically; cf.\ the proof of
Corollary \ref{pcor}.
Consequently, to prove Theorem \ref{ht} it is enough to prove the given expressions for $C_n^{(\mathbf k)}$. That is the goal of the remainder of \S \ref{ntcs}.

\subsection{Lattice translations  in the hyperbolic limit}

We need to understand the uniformization of 
translations by the lattice \eqref{lal} in the hyperbolic limit.
Consider the values $x(z+\gamma)$, where
$\gamma\in\Lambda $ and $(z,\tau)$ is changing in such a way that
$x=x(z,\tau)$ is fixed as $\zeta=\zeta(\tau)\rightarrow 0$. 
By  periodicity, we may reduce $\gamma$ by $\mathbb Z+\tau\mathbb Z$, leading
to twelve distinct values.
By Lemma \ref{hpl} and Lemma \ref{mll},   $x(z+1/2)$ behaves as $x/(2x-1)$ and one of the two branches
of $x(z\pm 1/3)$ as $1-x$. It will be convenient to write $x=(y+1)/2$, so that the corresponding transformations of $y$ are 
$y\mapsto 1/y$ and $y\mapsto-y$. Then, the
 other branch of $x(z\pm 1/3)$ behaves as
\begin{subequations}\label{xlh}
\begin{equation}\label{xph}\frac\zeta 2+\frac {y^2+3}{4(1-y^2)}\,\zeta^2+\mathcal O(\zeta^3). \end{equation}
By Lemma \ref{hpl},
\begin{equation}x\left(z+\frac\tau 2\right)=\frac{\zeta}{y+1}+\mathcal O(\zeta^2).\end{equation}
Combining these facts we find that one branch of
$x(z+\tau/2\pm 1/3)$ behaves as
\begin{equation}1+\frac{2y^2}{y^2-1}\,\zeta+\mathcal O(\zeta^2).\end{equation}
\end{subequations}
In conclusion, the twelve values  $x(z+\gamma)$ can be identified with 
the fixed point $(y+1)/2$, the three moving points \eqref{xlh}
and the additional three fixed and five moving points obtained from those four 
by replacing $y$ by $- y$ or $\pm y^{-1}$.

\subsection{The polynomials $P_j$ in the hyperbolic limit}
 
We will now investigate the limit  of the polynomials $P_j$, defined in \eqref{psb},  when the variable is fixed or specialized as in \eqref{xlh}.
Let  $f_j(y)=(1+y)^{[j/2]}(1-y)^{[(j-1)/2]}$. 
By a straightforward computation, we have as $\zeta\rightarrow 0$
\begin{subequations}\label{pjh}
\begin{align}\frac{2^{\left[\frac{3(j-1)}2\right]}P_j\left(\frac{y+1}2\right)}{(y+1)^{n-1}}&\rightarrow f_j(y),\\
\frac{2^{\left[\frac{3(j-1)}2\right]+1}(\hat\sigma P_j)\left(\frac{y+1}2\right)}{(y+1)^n}&\rightarrow(1-y)f_j(y)-2(1+y)f_j(-y),\end{align}
\begin{align}
\intertext{$\displaystyle\frac{2^{\left[\frac{3(j-1)}2\right]+2-n}y(1-y^2)^{n-1}P_j\left(
\frac\zeta 2+\frac {(y^2+3)\zeta^2}{4(1-y^2)}
\right)}{\zeta^{2n-2}}$}
&\rightarrow(1-y)f_j(y)-(1+y)f_j(-y),
\intertext{$\displaystyle\frac{2^{\left[\frac{3(j-1)}2\right]+1-n}(1-y^2)^{n}(\hat\sigma P_j)\left(\frac\zeta 2+\frac {(y^2+3)\zeta^2}{4(1-y^2)}\right)}{\zeta^{2n}}$}
&\rightarrow(1-y)f_j(y)+(1+y)f_j(-y).
\end{align}
Combining these relations with  Lemma \ref{pfl} 
gives
\begin{align}\MoveEqLeft
\frac{(-1)^{n+1}2^{2\left[\frac{j-1}{2}\right]+1-n}(y+1)^{2n-1}P_j\left(\frac\zeta{y+1}\right)}{\zeta^{\left[\frac j2\right]+n-1}}\rightarrow
f_{2n+1-j}(y)\\
\intertext{$\displaystyle\frac{(-1)^{n}2^{2\left[\frac{j-1}2\right]+1-n}(y+1)^{2n}(\hat\sigma P_j)\left(\frac\zeta{y+1}\right)}{\zeta^{\left[\frac j2\right]+n}}$}
&\rightarrow
(1-y)f_{2n+1-j}(y)-2(1+y)f_{2n+1-j}(-y)
\intertext{$\displaystyle\frac{(-1)^{n}2^{2\left[\frac{j-1}2\right]+3-2n}y(1-y^2)^{n-1}P_j\left(1-\frac{2y^2\zeta}{1-y^2}\right)}{\zeta^{\left[\frac j2\right]-1}}$}
&\rightarrow(1-y)f_{2n+1-j}(y)-(1+y)f_{2n+1-j}(-y),
\intertext{$\displaystyle\frac{(-1)^{n}2^{2\left[\frac{j-1}2\right]+1-2n}(1-y^2)^{n}(\hat\sigma P_j)\left(1-\frac{2y^2\zeta}{1-y^2}\right)}{\zeta^{\left[\frac j2\right]}}$}
&\rightarrow(1-y)f_{2n+1-j}(y)+(1+y)f_{2n+1-j}(-y).
 \end{align}
\end{subequations}

\subsection{Hyperbolic limit of the polynomials $T$}
\label{hts}

We will now consider the limit of the polynomials $T$, when
each variable is fixed or specialized as in \eqref{xlh}. That is, we
 consider the quantity
\begin{multline}\label{tes}
T\left(\frac{1+\mathbf s}2,
\frac\zeta 2+\frac {(\mathbf t^2+3)\zeta^2}{4(1-\mathbf t^2)},
\frac{\zeta}{\mathbf u+1},
1+\frac{2\mathbf v^2\zeta}{\mathbf v^2-1};\right.\\
\left.\frac{1+\mathbf w}2,\frac\zeta 2+\frac {(\mathbf x^2+3)\zeta^2}{4(1-\mathbf x^2)},\frac{\zeta}{\mathbf y+1},1+\frac{2\mathbf z^2\zeta}{\mathbf z^2-1}
\right),
\end{multline}
where $(1+\mathbf s)/2=((1+s_1)/2,\dots,(1+s_S)/2)$ and so on, with the total number of variables
$S+T+U+V+W+X+Y+Z=2N$.

We  express \eqref{tes} using the  determinant formula \eqref{sdd}.
Up to a  numerical factor, 
the leading term of the denominator is
\begin{multline*}\zeta^{2\binom T2+\binom U2+\binom V2+2\binom X2+\binom Y2+\binom Z2+TU+XY}\\
\begin{split}&\times\frac{\prod_{j=1}^S (1+s_j)^{T+U}(1-s_j)^V\prod_{j=1}^U(1-u_j)^T}{\prod_{j=1}^T(1-t_j^2)^{T-1}\prod_{j=1}^U(1+u_j)^{T+U-1}\prod_{j=1}^V(1-v_j^2)^{V-1}}\\
&\times\frac{\prod_{j=1}^W(1+w_j)^{X+Y}(1-w_j)^Z\prod_{j=1}^Y(1-y_j)^X}{\prod_{j=1}^X(1-x_j^2)^{X-1}\prod_{j=1}^Y(1+y_j)^{X+Y-1}\prod_{j=1}^T(1-z_j^2)^{Z-1}}\\
&\times\Delta(\mathbf s)\Delta(\mathbf t^2)\Delta(\mathbf u)\Delta(\mathbf v^2)\Delta(\mathbf w)\Delta(\mathbf x^2)\Delta(\mathbf y)\Delta(\mathbf z^2).
 \end{split}\end{multline*}
In the numerator, we replace each matrix element by its leading Taylor coefficient at $0$, which we obtain from \eqref{pjh}.
After reordering the rows and using 
$$\left[\frac{3(j-1)}2\right]-2\left[\frac{j-1}2\right]=\left[\frac j2\right],$$ this results in a numerical factor times
\begin{multline}\label{ed}\frac{\prod_{j=1}^S(1+s_j)^{N-1}}
{\prod_{j=1}^Tt_j(1-t_j^2)^{N-1}\prod_{j=1}^U(1+u_j)^{2N-1}\prod_{j=1}^Vv_j(1-v_j^2)^{N-1}}\\
\times\frac{\prod_{j=1}^W(1+w_j)^N}{\prod_{j=1}^X(1-x_j^2)^{N}\prod_{j=1}^Y(1+y_j)^{2N}\prod_{j=1}^Z(1-z_j^2)^{N}}\, \zeta^{2(N-1)T+(N-1)U-V+2NX+NY}\\
\times\det\left(\begin{matrix}f_j(s_i)\\
(1-w_i)f_j(w_i)-2(1+w_i)f_j(-w_i)\\
(1-t_i)f_j(t_i)-(1+t_i)f_j(-t_i)\\
(1-x_i)f_j(x_i)+(1+x_i)f_j(-x_i)\\
\zeta^{\left[\frac{j}2\right]}2^{\left[\frac{j}2\right]}f_{2N+1-j}(u_i)\\
\zeta^{\left[\frac{j}2\right]}2^{\left[\frac{j}2\right]}\big((1-y_i)f_{2N+1-j}(y_i)-2(1+y_i)f_{2N+1-j}(-y_i)\big)\\
\zeta^{\left[\frac{j}2\right]}2^{\left[\frac{j}2\right]}\big((1-v_i)f_{2N+1-j}(v_i)-(1+v_i)f_{2N+1-j}(-v_i)\big)\\
\zeta^{\left[\frac{j}2\right]}2^{\left[\frac{j}2\right]}\big((1-z_i)f_{2N+1-j}(z_i)+(1+z_i)f_{2N+1-j}(-z_i)\big)
\end{matrix}\right).\end{multline}
Here, each entry represents a block with  as many rows as the length of the corresponding vector variable and $2N$ columns.

Let us now introduce some notation.
For $k\in \mathbb Z_{\geq 0}$, $\mathbf s$, $\mathbf t$, $\mathbf u$, $\mathbf v$ vectors
and $p_j$ polynomials of degree $j-1$ we will write
\begin{align}\nonumber\psi^{(k)}(\mathbf s;\mathbf t;\mathbf u;\mathbf v)&\sim\frac 1{\prod_{j}u_j\Delta(\mathbf s)\Delta(\mathbf t)\Delta(\mathbf u^2)\Delta(\mathbf v^2)}\\
\label{psk}&\quad\times\det\left(\begin{matrix}p_j(s_i)\\
(1-t_i)^kp_j(t_i)-2(1+t_i)^kp_j(-t_i)\\
(1-u_i)^kp_j(u_i)-(1+u_i)^kp_j(-u_i)\\
(1-v_i)^kp_j(v_i)+(1+v_i)^kp_j(-v_i)
\end{matrix}\right).
 \end{align}
 The freedom to vary $p_j$ means that $\psi^{(k)}$ is defined only  up to 
a non-zero constant factor. It can be fixed  by choosing e.g.\ $p_j(x)=x^{j-1}$, but we prefer not to do so.
 By standard  arguments, $\psi^{(k)}$ is a symmetric polynomial in each of the four groups of variables.

We will consider three special cases of \eqref{ed}, the first one being
$U=V=Y=Z=0$. Since $f_j$ is a polynomial of degree $j-1$, \eqref{tes}
is equal to
\begin{equation}\label{tesa}
\frac{\prod_{j=1}^S(1+s_j)^{N-1-T}\prod_{j=1}^W(1+w_j)^{N-X}}{\prod_{j=1}^{T}(1-t_j^2)^{N-T}\prod_{j=1}^X(1-x_j^2)^{N-X+1}}\,\psi^{(1)}(\mathbf s;\mathbf w;\mathbf t;\mathbf x)\zeta^K+\mathcal O(\zeta^{K+1}),
\end{equation}
where
$K=(2N-T-1)T+(2N-X+1)X$.

Next, 
when  $S=T=W=X=0$, we may pull out the factor $\zeta^{\sum_{j=1}^{2N}[j/2]}=\zeta^{N^2}$ from \eqref{ed}, and conclude that \eqref{tes}    equals
\begin{multline}\label{tesb}\frac{1}{\prod_{j=1}^U(1+u_j)^{2N-U}\prod_{j=1}^V(1-v_j^2)^{N-V}\prod_{j=1}^Y(1+y_j)^{2N-Y+1}\prod_{j=1}^Z(1-z_j^2)^{N-Z+1}} \\
\times\psi^{(1)}(\mathbf u;\mathbf y;\mathbf v;\mathbf z)\zeta^K+\mathcal O(\zeta^{K+1}),\end{multline}
where now
$$K=N^2+(N-1)U-V+NY-\binom U2-\binom V2-\binom Y2-\binom Z2. $$
This also follows from \eqref{tesa} using the symmetry \eqref{sic}.

Finally, we consider \eqref{ed} when
\begin{equation}\label{oc}
U+V+Y+Z=2M+1\end{equation} is odd. We pull out the factor 
$\zeta^{[j/2]-M}$ from the $2M+1$ first columns, and $\zeta^M$ from the $2M+1$ last rows, giving in total a factor $\zeta^{M(M+1)}$. 
When $\zeta=0$, the resulting matrix is regular, with  the lower right
$(2M+1)\times(2N-2M-1)$ block vanishing. Thus, the determinant factors as the product of the lower left and the upper right block. 
 Again ignoring a numerical factor, we find that the determinant in \eqref{ed} equals
\begin{multline}\label{ptd}\zeta^{M(M+1)}\det_{2M+2\leq j\leq 2n}\left(\begin{matrix}f_j(s_i)\\
(1-t_i)f_j(t_i)-(1+t_i)f_j(-t_i)\\
(1-w_i)f_j(w_i)-2(1+w_i)f_j(-w_i)\\
(1-x_i)f_j(x_i)+(1+x_i)f_j(-x_i)
\end{matrix}\right)\\
\times
\det_{2n-2M\leq j\leq 2n}\left(\begin{matrix}
f_j(u_i)\\
(1-v_i)f_j(v_i)-(1+v_i)f_j(-v_i)\\
(1-y_i)f_j(y_i)-2(1+y_i)f_j(-y_i)\\
(1-z_i)f_j(z_i)+(1+z_i)f_j(-z_i)
\end{matrix}\right). \end{multline}

We now use that
$$f_{2M+1+j}(x)=(1-x)^M(1+x)^{M+1}f_j(-x). $$ 
Thus, the first determinant in \eqref{ptd}
is equal to
\begin{multline*}\prod_{j=1}^{S}(1+s_j)(1-s_j^2)^M\prod_{j=1}^T(1-t_j^2)^{M+1}
\prod_{j=1}^W(1-w_j^2)^{M+1}\prod_{j=1}^X(1-x_j^2)^{M+1}\\
\times\det_{1\leq j\leq 2n-2M-1}\left(\begin{matrix}f_j(-s_i)\\
f_j(-t_i)-f_j(t_i)\\
f_j(-w_i)-2f_j(w_i)\\
f_j(-x_i)+f_j(x_i)
\end{matrix}\right).
 \end{multline*}
Since  a similar identity holds for the second factor, \eqref{tes} is equal to
\begin{align}\notag\label{tese}\MoveEqLeft{\frac{\prod_{j=1}^s(1+s_j)^{M+N-T-U}(1-s_j)^{M-V}\prod_{j=1}^U(1-u_j)^{N-M-T-1}
}{\prod_{j=1}^T(1-t_j^2)^{N-M-T-1}\prod_{j=1}^U(1+u_j)^{M+N-T-U}\prod_{j=1}^V(1-v_j^2)^{M-V}}}\\
\notag &\times\frac{\prod_{j=1}^W(1+w_j)^{M+N-X-Y+1}(1-w_j)^{M-Z+1}\prod_{j=1}^Y(1-y_j)^{N-M-X}}{\prod_{j=1}^X(1-x_j^2)^{N-M-X}\prod_{j=1}^Y(1+y_j)^{M+N-X-Y+1}\prod_{j=1}^Z(1-z_j^2)^{M-Z+1}}\\
&\times \psi^{(0)}(\mathbf s;\mathbf w;\mathbf t;\mathbf x)\psi^{(0)}(\mathbf u;\mathbf y;\mathbf v;\mathbf z)\,\zeta^K+\mathcal O(\zeta^{K+1}),
\end{align}
where
\begin{multline*}K=2(N-1)T+(N-1)U-V+2NX+NY+M(M+1) \\
-2\binom T2-\binom U2-\binom V2-2\binom X2-\binom Y2-\binom Z2-TU-XY. 
\end{multline*}
This is reminiscent of Proposition \ref{tdtp}.
There is in general no such factorization when the left-hand side of \eqref{oc} is even. We stress that
\eqref{tesa}, \eqref{tesb} and
\eqref{tese} do \emph{not} always give the leading behaviour of \eqref{tes}, since  the polynomials $\psi^{(0)}$  and $\psi^{(1)}$ may vanish identically. Still, 
as we will see in \S \ref{tnkhs}, these identities
 together determine the leading behaviour of $T_n^{(\mathbf k)}$.

\subsection{The symmetric functions $\psi^{(k)}$}
\label{psks}

We will need some  properties of the symmetric polynomials
$\psi^{(k)}$ defined for $k\geq 0$ in \eqref{psk}. We 
extend the definition to $k\leq 0$ by 
\begin{align}\nonumber\psi^{(k)}(\mathbf s;\mathbf t;\mathbf u;\mathbf v)&\sim\frac{X(\mathbf s,\mathbf t,\mathbf u,\mathbf v)^k}{\prod_{j}u_j\Delta(\mathbf s)\Delta(\mathbf t)\Delta(\mathbf u^2)\Delta(\mathbf v^2)}\\
\label{pskn}&\quad\times\det\left(\begin{matrix}(1-s_i)^{-k}p_j(s_i)+2(1+s_i)^{-k}p_j(-s_i)\\
p_j(t_i)\\
(1-u_i)^{-k}p_j(u_i)-(1+u_i)^{-k}p_j(-u_i)\\
(1-v_i)^{-k}p_j(v_i)+(1+v_i)^{-k}p_j(-v_i)
\end{matrix}\right),
 \end{align}
where
$$X(\mathbf s,\mathbf t,\mathbf u,\mathbf v)=\prod_j(1+s_j)\prod_j(1+t_j)\prod_j(1-u_j^2)\prod_j(1-v_j^2). $$
For consistency, we  must  show that  \eqref{psk} and \eqref{pskn} agree 
 up to a constant factor when $k=0$.
To this end, we note that if $q_j(x)=p_j(x)-2p_j(-x)$, then
\begin{align*}
p_j(x)&=-\frac 13\big(q_j(x)+2q_j(-x)\big),\\
p_j(x)-p_j(-x)&=\frac 13\big(q_j(x)-q_j(-x)\big),\\
p_j(x)+p_j(-x)&=-\big(q_j(x)+q_j(-x)\big)
\end{align*}
Replacing  $p_j$ by $q_j$ in  \eqref{pskn} then gives the desired consistency. 

The following result shows that \eqref{pskn} is the natural extension of \eqref{psk}.

\begin{lemma}\label{pssl}
For $k\in\mathbb Z$,
 \begin{align}\label{pssa}\psi^{(k)}(\mathbf s;\mathbf t;\mathbf u;\mathbf v)\bigg|_{s_1=1}&\sim\psi^{(k+1)}(\hat{\mathbf  s};\mathbf t;\mathbf u;\mathbf v),
 \\
   \label{pssb}\psi^{(k)}(\mathbf s;\mathbf t;\mathbf u;\mathbf v)\bigg|_{t_1=1}
 &\sim X(\mathbf s;\hat{\mathbf  t};\mathbf u;\mathbf v)
 \psi^{(k-1)}(\mathbf s;\hat{\mathbf  t};\mathbf  u;\mathbf v)
 \end{align}
where $\hat{\mathbf  s}=(s_2,s_3,\dots)$.
\end{lemma}

\begin{proof}
Since \eqref{psk} and
\eqref{pskn} agree for $k=0$,  we may assume that  $\psi^{(k)}$ and $\psi^{(k\pm 1)}$  are either both given by \eqref{psk} or both by
\eqref{pskn}. We consider only the first case, the second case being similar.

To prove \eqref{pssa} for $k\geq 0$, we choose $p_j(x)=(1-x)^{j-1}$ in \eqref{psk}. Expanding the determinant along the first row, only the first entry contributes to the value at $s_1=1$. 
This leads to a factor $\prod_{i\geq 2}(s_i-1)^{-1}$ from the denominator, times \eqref{psk} with $\mathbf s$ replaced by $\hat{\mathbf  s}$ and $p_j(x)$ replaced by $p_{j+1}(x)=(1-x)p_j(x)$. Pulling out the factors $1-s_i$ from the corresponding rows  gives \eqref{pssa}.

To prove  \eqref{pssb} for $k\geq 1$, we choose instead 
$p_j(x)=(x+1)^{j-1}$.
In this case, we can pull out the factors $1+s_i$,  $1-t_i^2$, $1-u_i^2$ and
$1-v_i^2$ from the corresponding rows.  Cancelling the factor
$t_i-1$ appearing from the denominator yields \eqref{pssb}. 
\end{proof}

The following simple fact will be useful.

\begin{lemma}\label{opsl}
Suppose $\mathbf s$,  $\mathbf u$, $\mathbf v$ are vectors of length  $S$,  $U$, $V$  respectively,  Then, if $V=S+U+1$,
$$\psi^{(0)}(\mathbf s;-;\mathbf u;\mathbf v)\sim\prod_{i=1}^Ss_i\prod_{1\leq i<j\leq S}(s_i+s_j)\prod_{{1\leq i\leq S,\, 1\leq j\leq U}}(s_i^2-u_j^2) $$
and if $U=S+V-1$,
$$\psi^{(0)}(\mathbf s;-;\mathbf u;\mathbf v)\sim\prod_{1\leq i<j\leq S}(s_i+s_j)\prod_{{1\leq i\leq S,\, 1\leq j\leq V}}(s_i^2-v_j^2). $$
\end{lemma}

\begin{proof}
Let $p_j(s)=s^{j-1}$ in \eqref{psk}. In the first case, all matrix entries in the final $S+U+1$ rows and the $S+U$ columns with odd index are zero. Thus, the determinant factors as a constant times the two Vandermonde determinants
 $$\det\left(\begin{matrix}s_i^{2j-1}\\u_i^{2j-1}\end{matrix}\right)\det(v_i^{2j-2})=\prod_js_j\prod_ju_j\,\Delta(\mathbf s^2,\mathbf u^2)\Delta(\mathbf v^2).$$
The second case is proved similarly.
\end{proof}

We will need the following multivariable  binomial theorem.
Although it can be obtained, for instance, as a limit case of \cite[Cor.\ 7.24]{rs1}, we include a proof for completeness.

\begin{lemma}\label{mbl} One has
$$\Delta(\mathbf x)\prod_{j=1}^m\frac{(k+m-1)!}{(k+j-1)!}\,(1+x_j)^k=\sum_{\lambda_1,\dots,\lambda_m=0}^{k+m-1}\Delta(\boldsymbol\lambda)\prod_{j=1}^m\binom{k+m-1}{\lambda_j}x_j^{\lambda_j}. $$
\end{lemma}

\begin{proof}
Since $\Delta(\mathbf x)=\Delta(\mathbf x+1)=\det_{ij}((1+x_i)^{j-1})$, the left-hand side equals
\begin{align*}\MoveEqLeft\sum_{\sigma\in \mathrm S_m}\sgn(\sigma)\prod_{j=1}^m\frac{(k+m-1)!}{(k+j-1)!}(1+x_j)^{k+\sigma(j)-1}\\
&=\sum_{\sigma\in \mathrm S_m}\sgn(\sigma)\sum_{\lambda_1,\dots,\lambda_m}\prod_{j=1}^m\frac{(k+m-1)!}{(k+j-1)!}\binom{k+\sigma(j)-1}{\lambda_j}x_j^{\lambda_j}\\
&=\sum_{\lambda_1,\dots,\lambda_m}\det_{1\leq i,j\leq m}\left(\frac{(k+m-1-\lambda_i)!}{(k+j-1-\lambda_i)!}\right)\prod_{j=1}^m\binom{k+m-1}{\lambda_j}x_j^{\lambda_j}.
\end{align*}
Since the matrix elements are monic polynomials in $-\lambda_i$ of degree $m-j$, the final determinant equals $\det_{ij}((-\lambda_i)^{m-j})=\Delta(\boldsymbol\lambda)$.
\end{proof}

Note that $\psi^{(k)}$ is obtained from $\Delta(\mathbf x)\prod_j(1-x_j)^{|k|}\sim \det_{ij}((1-x_j)^{|k|}p_j(x_i))$ by applying operators
$f(x)\mapsto f(x)+\varepsilon f(-x)$ in each variable, where $\varepsilon\in\{0,\pm 1,\pm 2\}$. We can thus obtain explicit expressions for  $\psi^{(k)}$ from
Lemma~\ref{mbl} (with $x_j$ replaced by $-x_j$). We will only need this when the variables $t_j$ are absent. 

\begin{corollary} \label{psec}
Up to a  constant factor, the function
$$\psi^{(k)}(s_1,\dots,s_S;-;u_1,\dots,u_U;v_1,\dots,v_V) $$
can for $k\geq 0$ be expressed as
\begin{multline*}\frac 1{\prod_{j=1}^S(1-s_j)^k\prod_{j=1}^Uu_j\,\Delta(\mathbf s)\Delta(\mathbf u^2)\Delta(\mathbf v^2)}
\sum_{\substack{0\leq \lambda_1,\dots,\lambda_{S+U+V}\leq S+U+V+k-1\\ \lambda_{S+1},\dots,\lambda_{S+U}\text{\emph{ odd}}\\ \lambda_{S+U+1},\dots,\lambda_{S+U+V}\text{\emph{ even}}}}\Delta(\boldsymbol\lambda)\\
\times\prod_{j=1}^{S+U+V-1}\binom{S+U+V+k-1}{\lambda_j}\prod_{j=1}^S(-1)^{\lambda_j}s_j^{\lambda_j}\prod_{j=S+1}^{S+U}u_j^{\lambda_j}\prod_{j=S+U+1}^{S+U+V}v_j^{\lambda_j}
\end{multline*}
and for $k\leq 0$ as
\begin{multline*} \frac {\prod_{j=1}^S(1+s_j)^k\prod_{j=1}^U(1-u_j^2)^k\prod_{j=1}^V(1-v_j^2)^k}{\prod_{j=1}^Uu_j\,\Delta(\mathbf s)\Delta(\mathbf u^2)\Delta(\mathbf v^2)}
\sum_{\substack{0\leq \lambda_1,\dots,\lambda_{S+U+V}\leq S+U+V-k-1\\ \lambda_{S+1},\dots,\lambda_{S+U}\text{\emph{ odd}}\\ \lambda_{S+U+1},\dots,\lambda_{S+U+V}\text{\emph{ even}}}}\Delta(\boldsymbol\lambda)\\
\times\prod_{j=1}^{S+U+V-1}\binom{S+U+V-k-1}{\lambda_j}\prod_{j=1}^S(2+(-1)^{\lambda_j})s_j^{\lambda_j}\prod_{j=S+1}^{S+U}u_j^{\lambda_j}\prod_{j=S+U+1}^{S+U+V}v_j^{\lambda_j}.
\end{multline*}
\end{corollary}

Since the points
$$\xi_1=\frac\zeta 2-\frac{\zeta^2}{4}+\mathcal O(\zeta^3),\qquad
\xi_2=\frac\zeta 2+\frac{3\zeta^2}{4}+\mathcal O(\zeta^3)$$
correspond to  \eqref{xph} with $y=\infty$ and $y=0$, respectively, 
we are interested in limits of Corollary \ref{psec} when  the variables 
$u_j$ and $v_j$ tend to  $\infty$ or $0$.

\begin{corollary}\label{pcor}
Let $U=U_1+U_2$, $V=V_1+V_2$ and $N=S+U+V+|k|-1$. Then, the limit
\begin{equation}\label{psl}\lim_{u_1,\dots,u_{U_2},v_1,\dots,v_{V_2}\rightarrow \infty}\frac{\psi^{(k)}(s_1,\dots,s_S;-;0^{(U_1)},u_1,\dots,u_{U_2};0^{(V_1)},v_1,\dots,v_{V_2})}{\prod_{j=1}^{U_2}u_j^{2[(S-U+V+k)/2]}\prod_{j=1}^{V_2}v_j^{2[(S+U-V+k+1)/2]}} \end{equation}
always exists finitely. It is non-zero if and only if
\begin{equation}\label{uvc}\left|U-V+\frac 12\right|\leq S+|k|+\frac 12.\end{equation}
Moreover, when \eqref{uvc} is satisfied, \eqref{psl} can for
 $k\geq 0$ be expressed as
\begin{subequations}\label{psls}
\begin{multline}\frac 1{\prod_{j=1}^S(1-s_j)^k\Delta(\mathbf s)}
\sum_{0\leq \lambda_1,\dots,\lambda_{S}\leq N}\Delta(\boldsymbol\lambda)\\
\times\prod_{j=1}^{S}\left(\binom{N}{\lambda_j}(-1)^{\lambda_j}s_j^{\lambda_j}\prod_{\mu\in J}(\lambda_j-\mu)\right)
\end{multline}
and for $k\leq 0$ as
\begin{multline} \frac {\prod_{j=1}^S(1+s_j)^k}{\Delta(\mathbf s)}
\sum_{0\leq \lambda_1,\dots,\lambda_{S}\leq N}\Delta(\boldsymbol\lambda)\\
\times\prod_{j=1}^{S}\left(\binom{N}{\lambda_j}(2+(-1)^{\lambda_j})s_j^{\lambda_j}\prod_{\mu\in J}(\lambda_j-\mu)\right),
\end{multline}
where $J=J(N;U_1,U_2,V_1,V_2)$ is defined in \emph{\S \ref{hsrs}}.
\end{subequations}
\end{corollary}

\begin{proof}
We first observe that  \eqref{uvc} is equivalent to
$U\leq[(N+1)/2]$ and $V\leq [(N+2)/2]$, which is in turn equivalent to the existence of $U$ distinct odd and $V$ distinct even integers
in $[0,N]$. 

We apply the expressions given in Corollary \ref{psec}. By symmetry in the 
variables $u_j$, we may compute the limit in those variables by replacing
 $\Delta( \mathbf u^2)$ by $\prod_{j=1}^Uu_j^{2(j-1)}$ and 
choosing $(\lambda_{S+1},\dots,\lambda_{S+U})$ as the $U_1$ smallest and $U_2$ largest odd integers in $[0,N]$, written in increasing order. The analogous statement holds for the variables $v_j$.
It follows that \eqref{psl} is equal to \eqref{psls} when \eqref{uvc} holds and is otherwise identically zero.

It remains to show that \eqref{psls} cannot vanish identically, assuming  \eqref{uvc}. Since the non-zero terms are visibly linearly independent, it is
enough to show that such terms exist. The terms vanish
if and only if two summation indices are equal or one of them is in   $J$. Since $J$ has cardinality $N+1-U-V=S+|k|\geq S$, non-zero terms do exist.  
\end{proof}

\subsection{Leading behaviour of $T_n^{(\mathbf k)}$}
\label{tnkhs}

We will now apply the results of \S \ref{hts}--\ref{psks} to 
$$T_n^{(\mathbf k)}\left(\frac{1+s_1}2,\dots,\frac{1+s_m}2\right).$$
The denominator in \eqref{tnk} behaves as $D\zeta^J+\mathcal O(\zeta^{J+1})$, where
\begin{align*}
J&=2(k_1^-+k_2^-)(m+k_0^++k_1^++k_2^++k_3^+)+2(k_1^++k_2^+)(k_0^-+k_3^-),\\
D&\sim \prod_{j=1}^m (s_j+1)^{2k_0^-+2k_3^-}s_j^{2k_2^-}.
\end{align*}
We will identify the numerator as a specialization of \eqref{tes}. This can be done in several ways. For instance, if 
$k_3^+=1$, one of the variables from the left group should be specialized to $1$, which we can achieve either by specializing one of the variables  $s_j$ in
\eqref{tes} to $1$ or one of the variables $v_j$ to $0$. The choice must be made judiciously, so that the results of \S \ref{hts} yield  the correct leading term $C_n^{(\mathbf k)}\zeta^L$  (rather than $0\cdot\zeta^{L'}$ for some $L'<L$).

To be more precise, let $\mathrm C_1$ denote the subregime of C defined by $m=0$, $k_1\geq 0$ and $k_2\geq 0$ and let $\mathrm C_2=\mathrm C\setminus \mathrm C_1$. We  then express $T_n^{(\mathbf k)}$ in terms of \eqref{tes} with  parameters chosen as in the following table, where we write for short $i=k_0^++k_3^+$, $j=k_1^++k_2^+$, $k=k_0^-+k_3^-$, $l=k_1^-+k_2^-$.\\

\begin{tabular}{c|cccccccc}
 & $S$ & $T$ & $U$ & $V $ & $W$ & $X$ & $Y$ & $Z$\\
\hline
A & $m+i$ & $j$ & 0 & 0 & $k$ & $l$ & 0 & 0\\
 B & $m$ & $j$ & 0 & $i$ &0 & $j+m+1$ & $l-j-m-1$ & $k$\\
$\mathrm C_1$ & $0$ & $0$ & $j$ & $i$ &0 & $0$ & $0$ & $k$\\
$\mathrm C_2$ & $m$ & $l+m-1$ & $j+1-l-m$ & $i$ &0 & $l$ & $0$ & $k$
\end{tabular}

\vspace*{1ex}

The definition of the regimes guarantee that all entries  are non-negative. Note that the expression for $T$ 
used in regime $\mathrm C_2$ becomes negative in regime $\mathrm C_1$, which is the reason for treating $\mathrm C_1$ and $\mathrm C_2$ separately.

Choosing the parameters in \eqref{tes} as in
the table, we apply \eqref{tesa} in regime A, \eqref{tesb} in regime
 $\mathrm C_1$ and \eqref{tese} in regimes B and $\mathrm C_2$.
In the latter two cases, \eqref{oc} holds with
$M=k_0^++k_3^+-n-1$ and $M=k_0^++k_3^++k_1+k_2-n$, respectively. 
One may check that in each case $L=K-J$.  Thus, we obtain an expression
for the quantity $C_n^{(\mathbf k)}$ in Theorem \ref{ht} in terms of the polynomials 
$\psi^{(0)}$ or $\psi^{(1)}$. 

Consider first regime A. We specialize $\mathbf s=(s_1,\dots,s_m,1^{(k_0^++k_3^+)})$
and $\mathbf w=(1^{(k_0^-+k_3^-})$. By Lemma \ref{pssl}, our expression for
 $C_n^{(\mathbf k)}$ becomes
\begin{align*}C_n^{(\mathbf k)}&\sim \lim_{t_1,\dots,t_{k_1^+},x_1,\dots,x_{k_1^-}\rightarrow\infty}\frac{\prod_{j=1}^m(1+s_j)^{n-k_1-k_2-1}}
{\prod_{j=1}^ms_j^{2k_2^-}\prod_{j=1}^{k_1^+}t_j^{2(n-k_1-k_2)}\prod_{j=1}^{k_1^-}x_j^{2(n+1)}}\\
&\quad\times\psi^{(k_0+k_3+1)}(s_1,\dots,s_m;-;t_1,\dots,t_{k_1^+},0^{(k_2^+)};x_1,\dots,x_{k_1^-},0^{(k_2^-)}).
 \end{align*}
Applying  Corollary \ref{pcor}
and replacing $s_j$ by $2x_j-1$, 
 this reduces to  \eqref{ca}.

Next, we turn to regime B. Applying \eqref{tese} as indicated above,  the factor $\psi^{(0)}(\mathbf
s;-;\mathbf t;\mathbf x)$ is computed by Lemma \ref{opsl}.
In the factor $\psi^{(0)}(-;\mathbf y;\mathbf v;\mathbf z)$ we specialize all $y_j$ to $1$  and apply 
 Lemma \ref{pssl}. This leads to
\begin{align*}C_n^{(\mathbf k)}&\sim \prod_{j=1}^m\frac{(1+s_j)^{k_0+k_3-1-k_1-k_2}s_j^{2k_2+1}}{(1-s_j)^{n+1}}\prod_{1\leq i<j\leq m}(s_i+s_j)\\
&\quad\times\lim_{v_1,\dots,v_{k_0^+},z_1,\dots,z_{k_0^-}\rightarrow\infty} \frac{1}{\prod_{j=1}^{k_0^+} v_j^{2(n-k_0-k_3)}\prod_{j=1}^{k_0^-}z_j^{2(n+1)}}\\
&\quad\times\psi^{(k_1+k_2+m+1)}(-;-;v_1,\dots,v_{k_0^+},0^{(k_3^+)};z_1,\dots,z_{k_0^-},0^{(k_3^-)}),
\end{align*}
where the limit is  a non-zero constant by Corollary \ref{pcor}.
Note that
\eqref{uvc} reduces to $|k_0+k_3+1/2|\leq |k_1+k_2+m+1|+1/2$, which follows from the defining inequality for regime B.

The case $\mathrm C_1$ is treated in the same way as case A, but is simpler since $m=0$. Finally, the case $\mathrm C_2$ is treated similarly as case B.

\section{Comparison of notation}
\label{cns}

In this Section, we explain how  $T_n^{(\mathbf k)}$ are related to various polynomials appearing in  \cite{bm1, bm2, bh,fh,h,bm4,r,ras,zj}.

\subsection{Polynomials related to three-colour model}
\label{tcss}

In \cite{r}, we worked  with symmetric polynomials in $2n+1$ variables, defined by
\begin{multline}\label{sn}S_n(x_1,\dots,x_n,y_1,\dots,y_n,z)\\
=\frac{\prod_{i,j=1}^nG(x_i,y_j)}{\prod_{1\leq i<j\leq n}(x_j-x_i)(y_j-y_i)}\,
\det_{1\leq i,j\leq n}\left(\frac{F(x_i,y_j,z)}{G(x_i,y_j)}\right),
 \end{multline}
where $G$ is as in \eqref{gd} and 
$$F(x,y,z)=(\zeta+2)xyz-\zeta(xy+yz+xz+x+y+z)+\zeta(2\zeta+1).$$
Since $F(x,y,1)=2(x-\zeta)(y-\zeta)$, it follows that
$$T(x_1,\dots,x_{2n})=\frac{1}{2^n\prod_{j=1}^{2n}(x_j-\zeta)}\,S_n(x_1,\dots,x_{2n},1). $$
On the other hand, letting $y_n\rightarrow 1$, $x_n\rightarrow\zeta$ in 
\eqref{sn} one easily derives
$$\lim_{x_{2n}\rightarrow\zeta}\frac{S_n(x_1,\dots,x_{2n},1)}{\prod_{j
=1}^{2n}(x_j-\zeta)}=2^n\left(\zeta(\zeta+1)\right)^{n-1}S_{n-1}(x_1,\dots,x_{2n
-1}). $$
Combining these two results gives
\begin{equation}\label{ste}S_n(x_1,\dots,x_{2n+1})=\frac 1{(\zeta(\zeta+1))^n}\,T_{n+1}^{(0,0,0,0)}(x_1,\dots,x_{2n+1},\zeta).\end{equation}

As was mentioned at the end of \S \ref{tnks}, the right-hand side of \eqref{ste}
is essentially the function $T_{n}^{(0,0,0,-1)}$. To prove this, note that
$$P_j(\zeta)=\frac{\big(\zeta(\zeta+1)\big)^{n-1}}{2^{n+1}(\zeta-1)}\,(\hat\sigma P_j)(1) $$
for each $j$. 
Using this in  \eqref{sdd} gives
\begin{multline*}\prod_{j=1}^k(\zeta-x_j)T(x_1,\dots,x_k,\zeta;y_1,\dots,y_l)\\
=\frac{\big(\zeta(\zeta+1)\big)^{n-1}}{2^{n+1}(\zeta-1)}\prod_{j=1}^l(1-y_j)
T(x_1,\dots,x_k;y_1,\dots,y_l,1),
 \end{multline*}
where $k+l+1=2n$, which  leads to 
\begin{equation} \label{st}S_n(x_1,\dots,x_{2n+1})=\frac{2^n\prod_{j=1}^{2n+1}(x_j-\zeta)}{1-\zeta}\,T_n^{(0,0,0,-1)}(x_1,\dots,x_{2n+1})
. \end{equation}

Using \eqref{st}, we can rewrite 
 the polynomials $P_n$, $p_n$, $y_n$ and $\tilde p_n$
 of \cite[Prop.\ 8.1]{r} in terms of $T_n^{(\mathbf k)}$.
 Namely (recall that $\delta(n)=[n^2/4]$)
\begin{align}
\label{bpr}
P_n(x,\zeta)&=\frac{(-1)^{\left[n/2\right]}\left(\frac\zeta 2+1\right)^{n(n-1)-\delta(n-1)}(x-\zeta)}{(1-\zeta)\zeta^{n(n-1)}(\zeta+1)^{n(n-2)}\left(2\zeta+1\right)^{\delta(n-1)}}\,T_n^{(n,n,0,-1)}(x)\\
\label{pnr} p_n(\zeta)&=\frac{(-1)^{\left[n/2\right]}\left(\frac\zeta 2+1\right)^{n(n-1)-\delta(n-1)}}{(1-\zeta)\zeta^{n(n-1)}(\zeta+1)^{n^2-2n-1}\left(2\zeta+1\right)^{\delta(n)}}\,T_n^{(n+1,n,0,-1)},\\
\notag y_n(\zeta)&=\frac{(-1)^{\left[n/2\right]}\left(\frac\zeta 2+1\right)^{(n-1)^2-\delta(n-2)}}{2^{[(n+3)/2]}(1-\zeta)\zeta^{n(n-1)}(\zeta+1)^{n^2-2n-1}\left(2\zeta+1\right)^{\delta(n+1)}}\,T_n^{(n+2,n-1,0,-1)},
\\
\notag\tilde p_n(\zeta)&=
\frac{(-1)^{\left[n/2\right]+1}2^{[(n-1)/2]}\left(\frac\zeta 2+1\right)^{n^2-1-\delta(n)}}{(1-\zeta)\zeta^{n^2-1}(\zeta+1)^{n^2-2n-1}\left(2\zeta+1\right)^{\delta(n-1)}}\,T_n^{(n,n+1,0,-1)}.
\end{align}
The main result of \cite{r} is that the partition function of the three-colour model with domain wall boundary conditions can be expressed in terms of $p_n$ and $\tilde p_n$.

We mention that, in \cite{r}, there is a slight mistake in the proof that
$y_n$ is a polynomial. More precisely, the problem is to show 
 that $y_n(\zeta)$ is regular at $\zeta=0$. This follows from our Theorem \ref{ht}, which gives $T_n^{(n+2,n-1,0,-1)}=\mathcal O(\zeta^{n(n-1)})$.

\subsection{Polynomials of Bazhanov and Mangazeev}

We will now consider the 
polynomials $\mathcal P_n(x,z)$
 of Bazhanov and Mangazeev \cite{bm1, bm2, bm4}, which describe the ground state eigenvalue of Baxter's $Q$-operator for the supersymmetric ($\Delta=-1/2$) periodic XYZ chain of odd length.
In \cite{bm1},  these polynomials are defined up to a factor independent of $x$,
and then normalized  by writing
$$\mathcal P_n(x,z)=\sum_{k=0}^n r_k^{(n)}(z)x^k,$$
and requiring that 
$r_n^{(n)}(0)=1$. Since this only determines $\mathcal P_n(x,z)$ 
up to a multiplicative factor $f(z)$ with $f(0)=1$, we make the definition precise by requiring in addition that $\mathcal P_n(x,z)$ is not divisible by any non-constant polynomial in $z$. 
The following result will be proved in
 \S \ref{bmss} (using a result from \cite{r1b}).

\begin{proposition}\label{pprp}
The polynomials $\mathcal P_n$ and $P_n$ are related by
\begin{align*}
\mathcal P_n\left(y,\frac{\zeta}{(\zeta+2)(2\zeta+1)}\right)
&=\left(\frac{2}{(\zeta+2)(2\zeta+1)}\right)^{\delta(n)}\left(\frac{\zeta y+\zeta+2}{\zeta(\zeta+1)}\right)^n
\\
&\quad\times P_n\left(\frac{\zeta(y+2\zeta+1)}{\zeta y+\zeta+2},\zeta\right).
\end{align*}
\end{proposition}

Bazhanov and Mangazeev introduce the notation
\begin{align*}s_n(z)&=r_n^{(n)}(z)=\lim_{x\rightarrow\infty}\frac{\mathcal P_n(x,z)}{x^n}, \\
\bar s_n(z)&=r_n^{(0)}(z)=\mathcal P_n(0,z). 
\end{align*}
It follows from Proposition \ref{pprp} and \eqref{bpr} that
\begin{multline*} \left(\frac{(\zeta+2)(2\zeta+1)}{2}\right)^{\delta(n)}s_n\left(\frac{\zeta}{(\zeta+2)(2\zeta+1)}\right)\\
=\frac{(-1)^{\left[n/2\right]}\left(\frac\zeta 2+1\right)^{n(n-1)-\delta(n-1)}}{\zeta^{n(n-1)}(\zeta+1)^{n(n-1)}\left(2\zeta+1\right)^{\delta(n-1)}}\,T_n^{(n,n,0,0)},
\end{multline*}
\begin{multline*}
\left(\frac{(\zeta+2)(2\zeta+1)}{2}\right)^{\delta(n)}\bar s_n\left(\frac{\zeta}{(\zeta+2)(2\zeta+1)}\right)\\
=\frac{(-1)^{\left[n/2\right]+1}2^{n-1}\left(\frac\zeta 2+1\right)^{n^2-1-\delta(n-1)}}{\zeta^{n^2-1}(\zeta+1)^{n(n-1)}\left(2\zeta+1\right)^{\delta(n-1)}}\,T_n^{(n,n,1,-1)}.
\end{multline*}

\subsection{Polynomials of Zinn-Justin}
\label{zss}

In \cite{bm4}, Mangazeev and
Bazhanov gave a number of conjectures for eigenvectors
of  the supersymmetric XYZ Hamiltonian on a periodic chain
of odd length.
These involve polynomials $p_n$ (not to be confused with \eqref{pnr}) and $q_n$,
indexed by  $n\in\mathbb Z$, which can conjecturally be used to factorize 
the polynomials $s_n$ and $\bar s_n$. For instance, for $n\geq 0$ it is conjectured that
\begin{equation}\label{spp}s_{2n+1}(y^2)= p_n(y)p_n(-y).\end{equation}

 Zinn--Justin \cite{zj} expressed
 $p_n$ and $q_n$ in terms  of the symmetric
polynomials 
$$H_{2n}(x_1,\dots,x_n,y_1,\dots,y_n)
=\frac{\prod_{i,j=1}^nh(x_i,y_j)}{\prod_{1\leq i<j\leq n}(x_j-x_i)(y_j-y_i)}\,\det_{1\leq i,j\leq n}\left(\frac 1{h(x_i,y_j)}\right), $$
where
$$h(x,y)=1-(3+\zeta_{\text{Z}}^2)xy+(1-\zeta_{\text{Z}}^2)xy(x+y) $$
and $\zeta_{\text{Z}}$ is a parameter (with a subscript to distinguish it from our $\zeta$). 
 To see the connection to \eqref{tn}, we observe that
$$G\left(\phi(x),\phi(y)\right)=\frac{2\zeta^2(\zeta+1)^2}{(\zeta+2)^2}\,h(x,y), $$
where 
$$\phi(x)=\frac{\zeta}{\zeta+2}\left(1-2(\zeta+1)x\right)$$
and the parameters in $G$ and $h$ are related by
\begin{equation}\label{zzr}\zeta_{\text{Z}}=2\zeta+1.\end{equation}
With this relation between the parameters, it follows that
$$T_n(\phi(x_1),\dots,\phi(x_{2n}))\\
=\left(\frac{\zeta(\zeta+1)}{\zeta+2}\right)^{n(n-1)}\,H_{2n}(x_1,\dots,x_{2n}). $$

Using  identities from \cite[\S 4.1]{zj}
(where the factor
$\zeta^{2m(m-1)}$ should be replaced throughout by $\zeta^{m(m-1)}$)
we can now write
$$p_{n}\left(\frac 1{2\zeta+1}\right)=\frac{(-1)^nC_n(\zeta+2)^{n^2-n-1}}{\zeta^{n^2-2n-1}(\zeta+1)^{n(n-1)}(2\zeta+1)^{n^2+n+1}}\,T_{n}^{(-1,2n+1,0,0)}, $$
where $C_n=2^n$ for $n\geq 0$ and $C_n=3^{n+1}/2^{n+2}$ for $n\leq -1$. 
To be precise, for $n\geq 0$, this follows immediately from a corresponding identity in \cite{zj}, while for $n\leq -1$ we also need to apply  Corollary \ref{zscc}.
The reason for the different behaviour  of  $C_n$ 
for small and large $n$ is simply 
that the normalization chosen in
\cite{bm4} is not quite natural from our present perspective. 
Similarly,
$$q_{n}\left(\frac 1{2\zeta+1}\right)=D_n\left(\frac{\zeta+2}{\zeta(\zeta+1)(2\zeta+1)}\right)^{n(n+1)}T_{n+1}^{(0,2n+2,0,0)}, $$
where $D_n=1$ for $n\geq -1$ and $D_n=3^{n+2}/2^{2n+3}$ for $n\leq -2$.

The identity \eqref{spp}, and related conjectures from \cite{bm4},  can thus be expressed  in terms of  the polynomials $T_n^{(\mathbf k)}$. 
We hope to return to these conjectures in a subsequent paper in the present series.

\subsection{Proof of Proposition \ref{pprp}}
\label{bmss}

The starting point of \cite{bm1} is Baxter's TQ-equation for the 
eight-vertex model, which has the form
\begin{equation}\label{tq}T(u)Q(u)=\phi(u-\eta)Q(u+2\eta)+\phi(u+\eta)Q(u-2\eta). \end{equation}
Here, $T$ and $Q$ are eigenvalues of the transfer matrix and $Q$-operator, respectively, $\phi(u)=\vartheta_1(u|q)^{N}$ in the classical notation of \cite{ww}, 
and $u$, $q=\exp({\pi \ti \tau_{\text{BM}}})$ and $\eta$ are parameters of the model 
(we use the subscript  $\text{BM}$ to distinguish 
parameters used in \cite{bm1} from ours).
In the case $\eta=\pi/3$, $N=2n+1$, the ground state eigenvalue has the simple
form $T(u)=\phi(u)$. In this case, let 
$$f(u)=\phi(u)(Q_+(u)+Q_-(u)),$$ where
$Q_+$ and $Q_-$ are the two solutions to \eqref{tq}  defined in \cite{bm1}. 
Rewriting the defining properties of these solutions,
 it is straight-forward to check that 
if 
$$\tau_{\text{BM}}=2\tau,\qquad
u=2\pi \left(z+\frac {\tau+1}2\right)$$
 (so that $q=p^2$),
then the
function
$$g(z)=\frac{e^{6\pi \ti(n+1)z}f(u)}{\tha(-e^{6\pi\ti z};p^6)} $$
is an element of $\Theta_n^{(0,n,n,-1)}$.

In place of our uniformizing variables $(x,\zeta)$, 
Bazhanov and Mangazeev use $(x_{\text{BM}},z_{\text{BM}})$, which in our notation are given by
$$x_{\text{BM}}=-\frac{\tha(\om;p^2)^2\tha(e^{\pm 2\pi \ti z};p^2)}{\tha(-\om;p^2)^2\tha(-e^{\pm 2\pi\ti z};p^2)}, \qquad
z_{\text{BM}}=\left(\frac{\tha(\om;p^2)}{\tha(-\om;p^2)}\right)^4. $$
Using \eqref{xd} and  \cite[Lemma 9.1]{r}, one may check that
$$x_{\text{BM}}=\frac{(1-\zeta)(x-2\zeta-1)}{(1+\zeta)(x-1)},\qquad
z_{\text{BM}}=\frac{1+\zeta}{(1-\zeta)(1+2\zeta)}.$$

Rewriting \cite[Eq.\ (25)]{bm1} in our notation gives
\begin{align*}g(z)&\sim e^{-2\pi\ti z}\tha(e^{4\pi\ti z};p^2)\tha(\om pe^{\pm 2 \pi\ti z};p^2)^{3n-2}\frac{(x-\xi_1)^n(x-\xi_2)^n\left(x-\xi_3\right)^{n-1}}{x-\eta_3}\\
&\quad\times P_n(x_{BM},z_{BM}),
\end{align*}
up to factors independent of $z$. On the other hand, by
Theorem \ref{dt} and Proposition \ref{tnkp},
$$g(z)\sim e^{-2\pi\ti z}\tha(e^{4\pi\ti z};p^2)\tha(\om pe^{\pm 2 \pi\ti z};p^2)^{3n-2}\frac{(x-\xi_1)^n(x-\xi_2)^n}{x-\xi_3}\,T_n^{(0,n,n,-1)}(x).
 $$
It follows that 
\begin{equation}\label{psp}T_n^{(0,n,n,-1)}(x)\sim\frac{(x-1)^n}{x-\zeta}\,\mathcal P_n\left(\frac{(1-\zeta)(x-2\zeta-1)}{(1+\zeta)(x-1)},\frac{1+\zeta}{(1-\zeta)(1+2\zeta)}\right) \end{equation}
up to a factor independent of $x$.

Next we observe that, by Corollary \ref{sscc},
$$T_n^{(n,n,0,-1)}(x;\zeta)\sim
 T_n^{(0,n,n,-1)}
\left(\frac{(\zeta+2)x-(1+2\zeta)}{1-\zeta};-\zeta-1\right). $$
Combining this with \eqref{bpr} and \eqref{psp}, we find that
\begin{multline}\label{pp}\left(\frac{y\zeta+\zeta+2}{\zeta(\zeta+1)}\right)^n
P_n\left(\frac{\zeta(y+2\zeta+1)}{y\zeta+\zeta+2},\zeta\right)\\
=f(\zeta)
\left(\frac{(\zeta+2)(2\zeta+1)}{2}\right)^{\delta(n)}\mathcal P_n\left(y,\frac{\zeta}{(\zeta+2)(2\zeta+1)}\right),
 \end{multline}
where it remains to show that  $f\equiv 1$.

To proceed, we need the following result.

\begin{lemma}\label{ndpl}
The polynomial $P_n(x,\zeta)$ is not divisible by any non-constant polynomial in $\zeta$.
\end{lemma}

\begin{proof}
In \cite[Thm.\ 3.3]{r1b}, we give a partial differential equation 
for $T_n^{(\mathbf k)}$. In the special case $\mathbf k=(n,n,0,-1)$, it takes the form
\begin{equation}\label{le}\left(a(x,\zeta)\frac{\partial^2}{\partial x^2}+\frac{b(x,\zeta)}x\frac{\partial}{\partial x}+\frac{c(x,\zeta)}x+d(\zeta)\frac{\partial}{\partial \zeta}\right)P_n(x,\zeta)=0, \end{equation}
where $a$, $b$ and $c$ are explicit polynomials and
$$d(\zeta)=\zeta(\zeta-1)(\zeta+1)(\zeta+2)(2\zeta+1). $$
If  $P_n$ were divisible by some non-constant polynomial in $\zeta$, then 
we could write $P_n(x,\zeta)=(\zeta-\zeta_0)^NQ(x,\zeta)$, where $N>0$ and $Q(x,\zeta_0)\neq 0$. Inserting this into \eqref{le} gives 
$d(\zeta_0)=0$. But at the five zeroes of 
$d$,
it follows from \cite[Prop.\ 8.8]{r} that 
$P_n(x,\zeta_0)\neq 0$. 
\end{proof}

The differential equation \eqref{le} was conjectured in \cite{bm4}. By Proposition \ref{pprp}, it is in fact equivalent to the differential equation for $\mathcal P_n$ given without a complete proof in \cite{bm1}.

Let us now return to the function $f$ in \eqref{pp}. It is clear that $f$ can have poles only at the points $\zeta=0$, $\zeta=-1$, $\zeta=-2$ and $\zeta=-1/2$. Indeed, if $\zeta_0$ is any other pole, then $\mathcal P_n(y,z)$ would be divisible by $z-\zeta_0/(\zeta_0+2)(2\zeta_0+1)$, which contradicts our definition of $\mathcal P_n$. Similarly, by Lemma \ref{ndpl} and the fact that
$\zeta(y+2\zeta+1)/(y\zeta+\zeta+2)$ is independent of $y$ only for
$\zeta=0$ and $\zeta=1$,
  $f$ can have zeroes only at the points $\zeta=0$, $\zeta=1$, $\zeta=-2$ and  $\zeta=-1/2$. Moreover, by \cite[Eq.\ (8.4)]{r}
(or by \eqref{sia}), $f(\zeta)=f(1/\zeta)$. We conclude that
\begin{equation}\label{f}f(\zeta)=C\frac{\big((\zeta+2)(2\zeta+1)\big)^k(\zeta+1)^{2l}(\zeta-1)^{2m}}{\zeta^{k+l+m}},\end{equation}
where $C$ is a constant depending only on $n$ and $k$, $l$ and $m $ are integers, with $l\leq 0$ and $m\geq 0$.

It follows from \cite[Eq.\ (8.17)]{r}
that
$$\lim_{\zeta\rightarrow 0}f(\zeta)=\lim_{\zeta\rightarrow 0}\frac{2^nP_n\left({\zeta(y+1)}/2,\zeta\right)}{\zeta^n\mathcal P_n(y,0)}=\frac{\sum_{m=0}^n\binom{n+m}m(y+1)^{n-m}}{\mathcal P_n(y,0)}. $$
We conclude that $f$ is regular at $0$ and, in view of the normalization of $\mathcal P_n$, that $f(0)=1$.

To prove that $f$ is regular at $\zeta=-1$, we need to prove that
$$\lim_{\zeta\rightarrow -1}\frac{1}{(\zeta+1)^n}\,P_n(1+y(\zeta+1),\zeta) $$
exists finitely. 
Expressing $P_n$  in terms of  $T_n^{(0,n,n,-1)}$ and choosing $y=0$,  this is equivalent to the boundedness of $T_n^{(0,n,n,0)}/\zeta^{n(n-1)}$,
which is a special case of Theorem \ref{ht}.

Finally, to see that $f(1)\neq 0$, we must verify that $P_n(1,1)\neq 0$.
However, by \cite[Prop.\ 8.8]{r}, $P_n(1,1)$ is a non-zero constant times
$\phi_n(-1/3)$, where $\phi_n$ is as in \eqref{ph}. The non-vanishing of
$\phi_n(-1/3)$ follows recursively from \eqref{phr}. 

By the above considerations,
 $k=l=m=0$ in \eqref{f}. Since we have also showed that $f(0)=1$, $f$ is identically $1$. This completes the proof of Proposition~\ref{pprp}.

 \end{document}